\documentclass[prd,aps,amsfonts,superscriptaddress,nofootinbib,longbibliography,notitlepage,eqsecnum]{revtex4-1}
\usepackage{graphicx}
\usepackage{xcolor}
\usepackage{caption}
\usepackage{rotating}
\usepackage{dsfont}
\usepackage{amsmath,amssymb,graphics,amsthm,isomath}
\usepackage{subcaption}
\usepackage{physics }
\usepackage{circuitikz}
\usepackage[colorlinks=true, urlcolor=violet, linkcolor=blue, citecolor=red, hyperindex=true, linktocpage=true]{hyperref}
\usepackage[capitalise,compress]{cleveref}
\makeatletter
\def\doi@space{ }
\def\doi@guard#1{\ifx#1\doi@space\else#1\fi}
\def\doibase{https://doi.org/\doi@guard}
\makeatother
\usepackage{comment}

\usepackage{mathrsfs}

\newtheorem{thm}{Theorem}
\newtheorem{cor}[thm]{Corollary}
\newtheorem{lem}[thm]{Lemma}

\newtheorem{prop}[thm]{Proposition}

\theoremstyle{definition}
\newtheorem{defn}[thm]{Definition}
\newtheorem{nota}[thm]{Notation}

\theoremstyle{remark}
\newtheorem{rem}[thm]{Remark}

\makeatletter
\renewcommand{\p@subsection}{}
\renewcommand{\p@subsubsection}{}
\makeatother

\usepackage{xcolor}
\usepackage{mathtools}
\usepackage{tikz}

\usetikzlibrary{automata,positioning,backgrounds}
\usetikzlibrary{decorations.markings}

\newcommand{\mycomment}[1]{}

\usepackage{dsfont}
\begin{document}
\author{Andrew Osborne}
\affiliation{Department of Electrical and Computer Engineering, Princeton University, Princeton, NJ 08544, USA}
\author{Ciro Salcedo}
\affiliation{Princeton Quantum Initiative, Princeton University, Princeton, NJ 08544, USA}
\author{Andrew A. Houck}
\affiliation{Department of Electrical and Computer Engineering, Princeton University, Princeton, NJ 08544, USA}
\affiliation{Princeton Quantum Initiative, Princeton University, Princeton, NJ 08544, USA}
\title{Glued tree lattices with only compact localized states}
\begin{abstract}
    Flat-band physics is a central theme in modern condensed matter physics.
    By constructing a tight-binding single-particle system that has vanishing momentum dispersion in one or more bands, and subsequently including more particles and interactions, it is possible to study physics in strongly interacting regimes.
    Inspired by the glued trees that first arose in one of the few known examples of quantum supremacy,
    we define and analyze two infinite families of single-particle tight-binding models that have only flat bands, and only compact localized states despite having any nonnegative number of translation symmetries.
    The first class of model that we introduce is constructed by replacing a sufficiently large fraction of the edges in a generic countable graph with glued trees modified to have complex hoppings.
    The second class arises from thinking of complex-weighted glued trees as rhombi that can then be used to tile two-dimensional space, giving rise to the familiar dice lattice and infinitely many generalizations thereof, of which some are Euclidean while others are hyperbolic.
\end{abstract}
\maketitle
\tableofcontents
\section{Introduction}
A principal goal of traditional solid-state physics \cite{Ashcroft} is to understand the properties of large collections of atoms arranged in a regular and periodic fashion.
Crystalline material properties are determined, at least in part, by the manner in which constituent particles are arranged. Spatial symmetries of lattices provide simplifications that are powerful tools both formally and practically.
When a crystal lattice has some positive number of translation symmetries, its eigenvalues and eigenvectors may be described in terms of the same number of conjugate momenta. This is the (informal) statement of the lauded theorem of Bloch, a central pillar of condensed matter physics.

There is, however, more to say about materials than what follows from lattice structure.
When interactions dominate dynamics, exotic phases of matter may emerge.
There has been a great deal \cite{PhysRevLett.106.236803, PhysRevLett.106.236804,Jotzu_2014} of recent interest in strongly interacting phases of matter and their relation to lattice structure.
One of the challenges associated with experimental realizations of strongly interacting phases of matter is that Coulomb's law and the charge of an electron are fixed by nature; it is not possible to simply demand that physical particles interact more strongly than they used to.
Instead, one must contrive scenarios where ordinary interactions are strong compared to some other scale.

When an eigenenergy band is \textit{flat}, it is independent of any of the momenta generated by translation symmetries of a lattice. From a condensed matter perspective, this is of interest because any nonzero interactions will dominate dynamics. The strongly interacting dynamics enabled by flat-band physics are extensively studied in various forms including magic-angle twisted bilayer materials \cite{Andrei_2020,Cao_2018,doi:10.1073/pnas.1108174108}, kagome materials \cite{PhysRevResearch.4.023063, linegraphs,Wang_2023,Ye_2024}, and pyrochlore \cite{Huang_2024,PhysRevB.109.165147} materials. There are also recent studies devoted to simulating flat-band physics using superconducting circuits \cite{oliver,rhombic_exp,hyplatqed} and networks of resonators in circuit QED \cite{linegraphs,hyperblochmath}.

In this paper, we present several infinite families of lattices that have only compact localized eigenstates.
It is possible to construct such lattices in any number of Euclidean spatial dimensions
(and also non-Euclidean lattices) that have any number $n > 2$ of flat bands, and only flat bands
at the single-particle tight-binding level.
We construct these lattices using two different approaches.
In the first approach, we prove that a lattice with only compact localized states can be realized by replacing a sufficiently  large number of edges on a countable graph with a complex-weighted glued tree, thereby generalizing the
rhombic lattice\footnote{This lattice has also been called the Rhombi-chain lattice\cite{rizzi}}\cite{PhysRevLett.88.227005,rizzi,kremer2018nonquantizedsquareroottopologicalinsulators}, which has been studied for some time.
A second approach arises from viewing glued trees as rhombi that partition regular polygons that are then arranged into tessellations of space; the dice lattice \cite{dicehaldane,debnath2024magnonsdicelatticetopological,majorana,PhysRevB.101.235406} is an example that arises from this construction.
The mechanism that yields this very special eigenstate structure is the destructive interference of paths of equal length across plaquettes, or Aharonov--Bohm caging.

\subsection{An outline of the paper and summary of results}

In an effort to appeal to mathematical audiences, we have elected for a presentation tending toward the rigorous.
On the other hand, we aim to be pedagogical and self-contained, so we include explicitly a number of definitions and intermediate steps that could reasonably be omitted.

In this paper, our ultimate goal is to introduce and study many families of Hamiltonian models with eigenstates that, despite having any number of translation symmetries, are purely and perfectly local.
In a few words, the mechanism that we exploit to this end is Aharonov--Bohm caging, whereby a particle is confined to some region because the paths that leave the region must destructively interfere.
A set of $n > 1$ complex numbers $e^{\mathrm i \theta_i}$ of unit norm are said to ``destructively interfere" if
\begin{equation}\label{eqn:destruct}
    \sum_{i = 1}^n e^{\mathrm i \theta_i} = 0.
\end{equation}
The simplest example of a lattice model that has only compact localized eigenstates caused by destructive interference is
\begin{equation}\label{eqn:trivial}
    H = \delta \sum_{i \in \mathbb Z} |i \rangle \langle i +1 |+ \text{h.c.}
\end{equation}
with $\delta = \sum_j e^{\mathrm i \theta_j}$.
If $\delta = 0$, then $H = 0$, and any vector is a null vector of $H$.
Of course, this example is too simple to be of any real interest.
Nonetheless, the mechanism whereby eigenstates are confined in space by (\ref{eqn:trivial}) is transparent:
particles cannot travel any distance whatsoever and thus have no dispersion.
The simplest nontrivial extension of this idea is to
add $n$ disjoint paths between site $i$ and site $i+1$, each with its own new vertex.
This construction gives rise to a quasi-one-dimensional lattice which we will study extensively in the main sections of this paper.

The motif that will allow us to create lattices that centrally feature destructive interference is the so-called glued tree \cite{childs}, which arises as an example of constructive interference separating classical and quantum computation.
Heuristically, glued trees are a family of graphs defined by the manner in which they are constructed. The term ``glued tree" refers to a graph made from a pair of trees that have been glued together at the leaves.
We find that chains of glued trees, or lattices arising from tilings of space with regular polygons made of glued trees can be constructed to have only compact localized eigenstates.
By understanding the connection between intrinsically one-dimensional glued tree lattices like the rhombic lattice and higher dimensional lattices like the dice lattice \cite{dicehaldane,majorana,PhysRevB.101.235406,debnath2024magnonsdicelatticetopological}, we are able to generalize both.

We begin in section \ref{sec:warmup} with a simple and well-known example of a model with only compact localized eigenstates.
Then, we introduce central definitions in Section \ref{sec:gluedtrees}.
In particular, we make a precise definition of a glued tree, along with discussions relating complex-weighted adjacency matrices to
tight-binding models, and we argue that such matrices are a natural object of study both physically and mathematically.

In section \ref{sec:lattices}, we begin studying infinite lattices constructed from the motifs defined in Section \ref{sec:gluedtrees}.
We solve the adjacency matrix eigenvalue problem explicitly for many glued trees. For the glued trees that do not seem to admit a simple eigenvalue formula, we provide instead an algorithm for finding eigenvalues that runs in a time logarithmic in the number of vertices.
This result can be achieved either by manipulating a characteristic polynomial or by explicitly constructing eigenstates.
In section \ref{sec:abstates}, we prove that destructive interference can result in the perfect confinement of particles to finite regions in a lattice, and that this can happen at a great variety of different configurations, including points where ``magnetic flux" takes on some arbitrarily small value. Again, this result can be achieved either by studying a characteristic polynomial or by manipulating the resolvent of the adjacency matrix.
Then, we go on to prove our main theorem in section \ref{sec:confine} which says, intuitively, that the eigenstates of a Hamiltonian must all be compact localized if a particle initialized to an arbitrary site on a lattice is perfectly confined to a finite radius.

Finally, we introduce many lattices that can be shown by Theorem \ref{thm:main} to have only compact localized eigenstates in section \ref{sec:examples}, beginning with a recipe for constructing any number of such lattices from any countable graph.
Among the examples we introduce are several infinite families of model that are Euclidean or hyperbolic two-dimensional generalizations of the dice lattice with the additional property that no net magnetic field is required to localize eigenstates.
The models of this kind that we introduce can be easily extended to any $\{p,3\}$ or $\{2n, 2q\}$ tiling of space with $p \geq 6$, $n \geq 2$, and sufficiently large $q$.
The arguments in the main text are largely unreliant on Bloch's theorem, allowing for simple theorems that apply to hyperbolic and Euclidean (as well as aperiodic) graphs alike.

The appendices contain alternate proofs to (some slightly less general versions of) several theorems in the main text.

\subsection{A warm-up example}\label{sec:warmup}

\begin{figure}
    \centering
    \begin{subfigure}[c]{0.40\linewidth}
    \centering
    \includegraphics{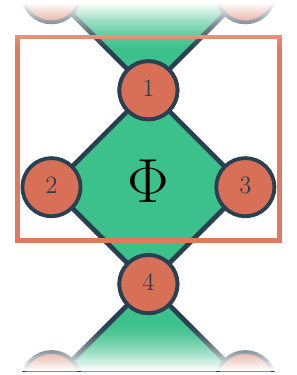}
    \end{subfigure}
    \begin{subfigure}[c]{0.58\linewidth}
    \centering
    \includegraphics{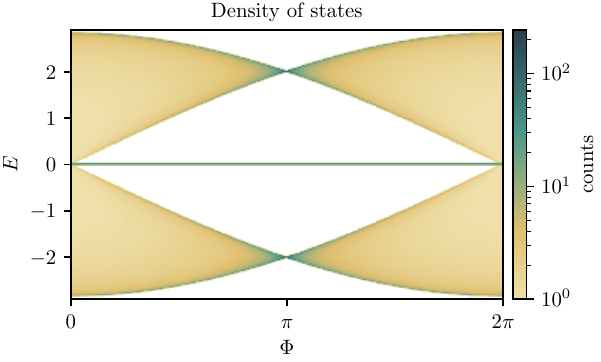}
    \end{subfigure}
    \caption{\textbf{Left: } The rhombic lattice, a quasi-one-dimensional lattice that has all flat bands at $\Phi = \pi$. An orange box encloses a single unit cell. The single-particle rhombic lattice Hamiltonian eigenenergies are the eigenvalues of a matrix in $\mathbb C^{3 \times 3}$.
    \textbf{Right: }  The density of states of the rhombic lattice at various values of $\Phi$. The fact that the colored points are finite in number at the vertical slice at $\Phi = \pi$ means that the bands of the rhombic lattice are flat at $\Phi = \pi$. Another fact that can be inferred from the figure above is that the spectrum at $\Phi = 0$ is gapless, and that the largest spectral gap at any value of $\Phi$ occurs at $\Phi = \pi$. The horizontal line at $E = 0$ means that there is always a flat band at zero energy.
    }
    \label{fig:rhombic}
\end{figure}

Consider the graph drawn in Figure \ref{fig:rhombic}. The lattice depicted is the rhombic lattice \cite{rhombic_exp,kremer2018nonquantizedsquareroottopologicalinsulators,PhysRevLett.88.227005}, which is of some theoretical and experimental interest.  Circles represent  sites, edges represent complex hopping terms of unit norm, and plaquettes are shaded to indicate that there is a flux $\Phi$ piercing such faces.
The phases associated with hopping terms are chosen to satisfy the fluxes indicated by shading.
The hopping part of the Bose--Hubbard Hamiltonian for this lattice, as realized in the superconducting circuit experiment of \cite{rhombic_exp}, is given without loss of generality by
\begin{equation}\label{eqn:rhombicham}
    H = \sum_{n = -\infty}^\infty  \left(e^{\mathrm i \phi_{12} } a^\dagger_{3 n + 1} a_{3n + 2} +  e^{\mathrm i \phi_{13}} a^\dagger_{3n + 1} a_{3 n + 3}  + e^{\mathrm i \phi_{24}} a^\dagger_{3n + 2} a_{3n + 4} +e^{\mathrm i \phi_{34}}a^{\dagger}_{3n + 3} a_{3n+4} \right) + \text{h.c.}
\end{equation}
where the sites $3n+1$, $3n+2$, $3n+3$ of the $n$th unit cell are labeled $1$, $2$, $3$ in Figure \ref{fig:rhombic} and site $3n+4$ is the site labeled $1$ in the next cell (labeled $4$ in the figure), with $\phi_{ij} = -\phi_{ji}$, and
\begin{equation}\label{eqn:plaflux}
    \phi_{12} + \phi_{24} + \phi_{43} + \phi_{31} = \Phi.
\end{equation}
The fact that (\ref{eqn:plaflux}) is underconstrained corresponds to a gauge freedom, which is unimportant for our purposes.
We are interested in the properties of this Hamiltonian restricted to the subspace of states $|\psi\rangle$ such that
\begin{equation}
    \sum_{i = -\infty}^\infty N_i |\psi\rangle = |\psi\rangle,
\end{equation}
where $N_i = a_i^\dagger a_i $ is the number operator. Such states will be called single-particle states, and they of course span a subspace of the full Hilbert space of states on which (\ref{eqn:rhombicham}) may act.
The Hamiltonian (\ref{eqn:rhombicham}) restricted to the single-particle sector is called the \textit{single-particle} Hamiltonian, and it is most simply recognized as the adjacency matrix of the underlying lattice, together with nonzero phases along each edge.
Explicitly
\begin{equation}\label{eqn:single}
    H_{\text{single-particle}} =
    \sum_{\langle i,j \rangle} e^{\mathrm i \phi_{ij}} |i \rangle \langle j | + e^{\mathrm i \phi_{ji}} |j\rangle \langle i |.
\end{equation}
Fourier transforming or equivalently applying Bloch's theorem (\ref{eqn:single}) yields the so-called Bloch Hamiltonian:
\begin{equation}
    H(k) = \begin{pmatrix}
        0 & e^{\mathrm i \phi_{12}} + e^{\mathrm i \phi_{24}}e^{\mathrm i k} & e^{\mathrm i \phi_{13}} + e^{\mathrm i \phi_{43}}e^{\mathrm i k} \\
        e^{\mathrm i \phi_{21}} + e^{\mathrm i \phi_{42}} e^{-\mathrm i k} & 0 & 0 \\
        e^{\mathrm i \phi_{31}} + e^{\mathrm i \phi_{34} } e^{-\mathrm i k} & 0 & 0
    \end{pmatrix}.
\end{equation}
After choosing a gauge where $\phi_{12} = \Phi$ and all other $\phi_{ij}$ are zero, we find
\begin{equation}
    \begin{pmatrix}
        0 & e^{\mathrm i \Phi} + e^{\mathrm i k} & 1 + e^{\mathrm i k} \\
        e^{-\mathrm i \Phi} + e^{-\mathrm i k} & 0 & 0 \\
        1 + e^{-\mathrm i k} & 0 & 0
    \end{pmatrix}.
\end{equation}
Very straightforwardly, we find energy eigenvalues to be
\begin{equation}
    E_i(k) = \tau_i  \sqrt{2} \sqrt{2 + \cos(k) + \cos(k - \Phi)}
\end{equation}
with $\tau_i \in \{1,0,-1\}$.
Of course, if $\Phi = \pi$, $E_i(k)$ is completely independent of $k$, regardless of the value of $\tau_i$.

From a physics perspective, it is rather intuitive that the rhombic lattice at $\pi$ flux yields all flat bands, since the paths across a plaquette destructively interfere \`a la Aharonov--Bohm effect, but despite this intuition, mathematical generalizations are nontrivial.
Nonetheless, our goal in this paper is to produce the broadest family of lattices with the property that a correctly chosen plaquette flux makes  \emph{all} energy eigenvalues of the relevant single-particle Hamiltonian independent of $k$ as we can.

\section{Graph theory and glued trees}\label{sec:gluedtrees}
We expect that many of the following definitions are commonly known, but we include them so that this paper is self-contained nonetheless.
We reserve a bold definition heading for definitions that are either nonstandard or of particular importance.
A \textit{graph} is a tuple $(E,V)$ with $V \subset \mathbb Z$, and $E \subset V \times V$.
All of the graphs of interest in this paper may safely be regarded as undirected.
The elements of $V$ are called \textit{vertices}, and the \textit{size} of a graph is the number of vertices in its vertex set.
Requiring that the vertices of graphs are integers is no loss of generality since we only consider graphs with a countable number of vertices.
A \textit{path} between a pair of vertices is a sequence of edges that starts on  one vertex and ends on the other.
The \textit{distance} between a pair of vertices is the smallest number of edges that can be traversed in order to travel from one vertex to the other.
The \textit{diameter} of a graph is the largest distance between any pair of vertices in that graph.
The \textit{degree} of a vertex in a graph is the number of edges connected to that vertex.
\begin{figure}
    \centering
    \includegraphics{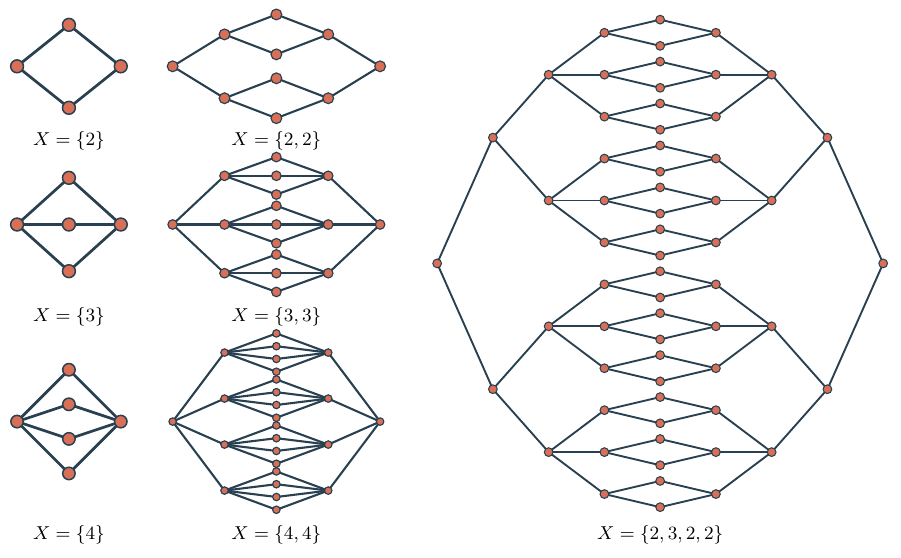}
    \caption{Various examples of glued trees, each labeled by the sequence $X$ from which it is grown; the roots $f_{A^X}$ and $l_{A^X}$ are drawn at the left and right. \textbf{Left:} the $p$-nary trees with $p = 2, 3, 4$ (rows) of depth $d = 1$ (shrubs) and $d = 2$ (columns). \textbf{Right:} a glued tree of depth four grown from the nonconstant sequence $X = \{2,3,2,2\}$. The order of entries in  $X$  is important: $\{2,3,2,2\}$ and $\{2,2,3,2\}$, for instance, grow different trees.}\label{fig:tree_examples}
\end{figure}
The \textit{adjacency matrix} of a graph $(E,V)$ is a linear operator $A$ on $\ell^2(V)$ such that
\begin{equation}
    \langle u | A | v \rangle = \begin{cases}
        1 & \{u,v\} \in E \\
        0 & \{u,v\} \not\in E.
    \end{cases}
\end{equation}
When we refer to the spectrum of a graph, we refer to the eigenvalue spectrum of the graph's adjacency matrix.
A graph is determined by its adjacency matrix and vice versa. We will use $\mathcal A(G)$ to mean the adjacency matrix of the graph $G$, and likewise $\mathcal G(A)$ to mean the graph whose adjacency matrix is $A$.

\begin{defn}[$p$-shrub]
    Fix $p > 0$ to be an integer.
    Let
    \begin{equation}
        V = \{0,1,2,\dots,p+1\},
    \end{equation}
    and
    \begin{equation}
        E   = \{ \{0,i\}\,:\, 1 \leq i \leq p\} \cup \{(i,p+1) \,:\, 1 \leq i \leq p\}
    \end{equation}
    is called the $\mathbf p$\textbf{-shrub}. We will use the symbol $S_p$ to denote the $p$-shrub.
\end{defn}
The $p$-shrub is also precisely the $(2,p)$ complete bipartite graph, $K_{2,p}$.
The adjacency matrix of $S_p$ is given by
\begin{equation}
    \mathcal A(S_p) = \begin{pmatrix}
        0 & 1 & 1 & \dots & 1& 1& 0 \\
        1 & 0 & 0 & \dots & 0 & 0 & 1 \\
        1      & 0 & 0 & \dots & 0 & 0 & 1 \\
        \vdots & \vdots  & \vdots  &  \ddots     & \vdots  &  \vdots & \vdots \\
        1      & 0 & 0 & \dots & 0 & 0 & 1 \\
        1      & 0 & 0 & \dots & 0 & 0 & 1 \\
        0 & 1 & 1 & \dots & 1& 1& 0\\
    \end{pmatrix}.
\end{equation}

\begin{nota}
    Let $B \in \mathbb C^{n \times n}$ be a square matrix. We write $|f_B\rangle$ to mean the element of $\mathbb R^{n}$ whose first entry is equal to one and all others are equal to zero.
    Likewise $|l_B\rangle$ is the vector in the same space whose last entry is equal to one and all others are equal to zero.
\end{nota}
Our reason for choosing this notation is that $B$ and $|f_B\rangle$ are compatible, for any square matrix $B$.
\begin{nota}
    Fix some positive integer $n$. We write $\mathbb I_n$ to mean the $n$ by $n$ identity matrix, and we write $|\mathds 1_n\rangle$ to be the vector in $\mathbb R^n$ whose entries are all equal to one.
\end{nota}

\begin{defn}[Growth] \label{defn:growth}
    Fix $d$ to be a positive integer, and let $X = \{x_1, x_2 , \dots x_d\} $ be a sequence of positive integers such that $x_d$ is greater than one.
    Let $Y_1, Y_2, \dots, Y_d$ satisfy the partitioned matrix recurrence relation
    \begin{equation}\label{eqn:growth}
        \begin{aligned}
            Y_{i} &= \begin{pmatrix}
                0 & \langle \mathds 1_{x_i} | \otimes \langle f_{Y_{i-1}} |   & 0 \\
                |\mathds 1_{x_i} \rangle \otimes | f_{Y_{i-1}}\rangle  & \mathbb I_{x_i} \otimes Y_{i-1} & |\mathds 1_{x_i}\rangle \otimes | l_{Y_{i-1}}\rangle \\
                0 & \langle \mathds 1 _{x_i} | \otimes \langle l_{Y_{i-1}} | & 0
            \end{pmatrix} \\
            Y_1 &= \mathcal A (S_{x_1})
        \end{aligned}
    \end{equation}
    where $\otimes$ is the Kronecker product.
    The recurrence relation (\ref{eqn:growth})  is called the \textbf{growth relation} determined by $X$. We say that $Y_d$ is the matrix \textbf{grown} by $X$, and we will write $A^X$ to denote the matrix grown by $X$. Explicitly, here $A^X := Y_d$.
\end{defn}
We forbid $x_d = 1$ with no loss of generality, as we will see in Section \ref{sec:examples}.

\begin{defn}[Glued tree]\label{defn:gluedtree}
   Let $G$ be a graph. If there exists a sequence of numbers $X = \{x_1, x_2, \dots, x_d\}$  of length $d > 0$ such that $x_d > 1$, and $G = \mathcal G(A^X)$, we say that $G$ is a \textbf{glued tree}. We use the term \textit{depth} of $G = \mathcal G( A^X)$ to mean $|X|$.
\end{defn}

Glued trees are precisely the graphs that may be grown from shrubs.
Every glued tree may be produced by the following algorithm
\footnote{The procedure below is not precisely the same as Definition \ref{defn:growth} since we haven't been careful about which vertices should be regarded as ``roots". Doing so is possible, but we regard this possibility as unnecessarily technical. }, given an appropriate sequence $X$:
\begin{enumerate}
    \item Define $Y_1$ to be $S_{x_1}$, the $x_1$-shrub
    \item Make $x_i$ copies of $Y_{i-1}$, together with an additional pair of vertices that are initially isolated
    \item Choose one of the vertices of degree $x_{i-1}$ (the ``roots'') in each copy of $Y_{i-1}$, and connect them to a single isolated vertex
    \item Connect the remaining isolated vertices to the ``roots" of degree $x_{i-1}$ that were not chosen in the previous step
    \item Repeat until $X$ is exhausted.
\end{enumerate}

We have chosen the name ``glued tree" in part because this is the name that was chosen in \cite{childs}, where these graphs arose as a construction central  to one of the few known examples of an exponential separation between classical and quantum computation.
Another reason for this name is that glued trees can always be constructed by building a tree, making a copy of the tree, and gluing the vertices of the two copies together.
We also find it amusing to think of trees being things that may be grown from shrubs.
Figure \ref{fig:tree_examples} contains various examples of glued trees, including one grown from a sequence $X$ whose entries are not all equal; another such tree appears, as the unit cell in a lattice, in Figure \ref{fig:1d}.

\begin{defn}[$p$-nary]\label{defn:pnary}
    Let $G = \mathcal G(A^X)$ be a glued tree. We say that $G$ is $\mathbf p$\textbf{-nary} if the elements of $X = \{x_1, x_2, \dots, x_d\}$ satisfy
    \begin{equation}
        p = x_1 = x_2 = \dots = x_d.
    \end{equation}
\end{defn}
\begin{rem}
    Every glued tree of depth one is $p$-nary.
\end{rem}
In other words, a glued tree of depth one is simply a shrub. Equivalently, every $p$-shrub is a depth-one $p$-nary glued tree.

\begin{figure}
    \includegraphics[width=1\linewidth]{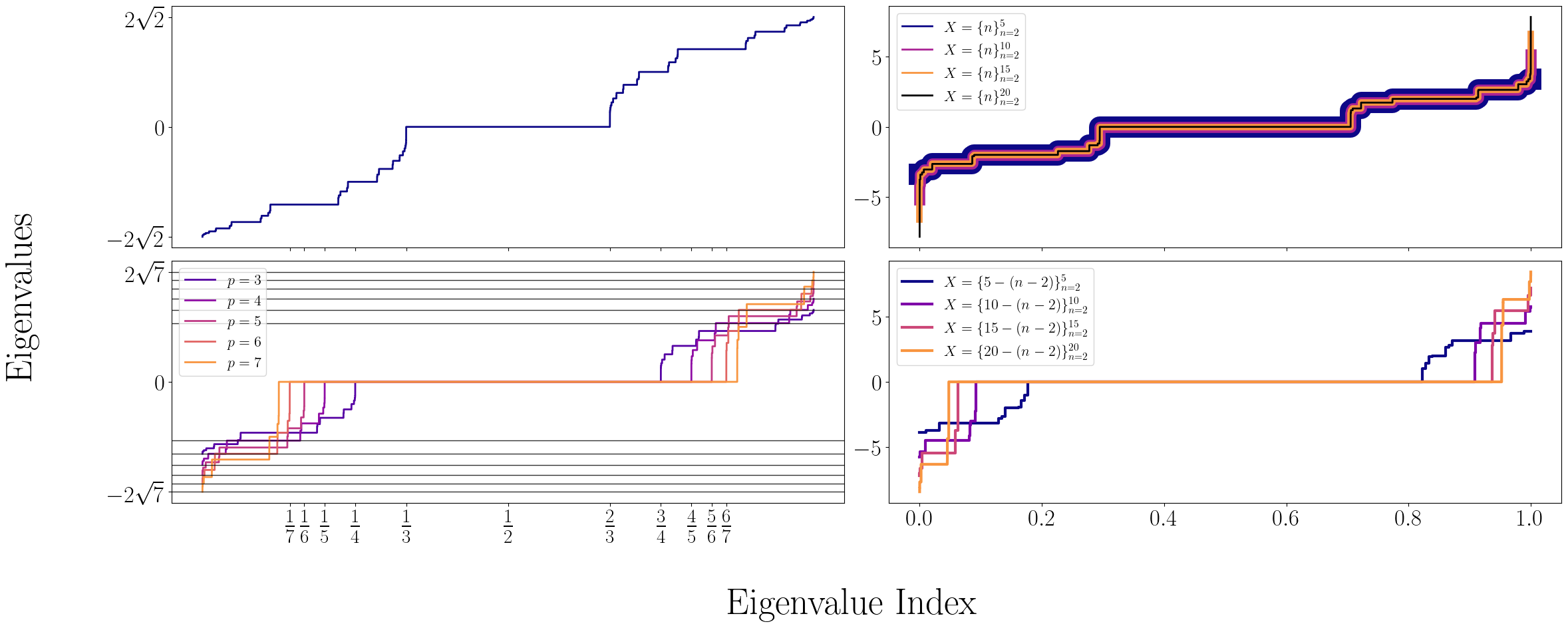}
    \caption{Eigenvalues of various glued trees sorted in increasing order. Horizontal axes are scaled by the total number of eigenvalues.
    Notice that approximately $\frac{p-1}{p+1}$ of the total eigenvalues of the $p$-nary glued trees are zero.
    \textbf{Top Left:} The eigenvalues of the binary ($2$-nary) glued tree of depth 2000. We remark that the curve depicted bears a striking resemblance to the Cantor ternary function. \textbf{Bottom Left:} Eigenvalue spectra of $p$-nary trees of depth 2000 for various values of $p$. The horizontal lines  denote the largest eigenvalues of each tree, which approach $2 \sqrt p$. \textbf{Top Right:} The eigenvalue spectra of upward cascading trees. Note that $\{n\}_{n = 2}^5 = \{2,3,4,5\}$, which is a notation defining a sequence. \textbf{Bottom Right:} Eigenvalue spectra of downward cascading trees.
    Note that the degeneracy of the zero eigenvalue appears to be fixed by $x_1$ in the sense that approximately $\frac{x_1 - 1}{x_1 + 1}$ of eigenvalues are zero.
    }
    \label{fig:eigvals}
\end{figure}
\begin{thm}\label{thm:fluxless}
    Let $X = \{x_1, x_2, x_3, \dots, x_d\}$ be a sequence of positive integers with $x_i > 1$ for all $i = 1,2,\dots,d$.
    For $i \geq 1$, let
\begin{equation}
H_i = \begin{pmatrix}
0 & \sqrt{x_i} \langle f_{H_{i-1}}| & 0 \\
\sqrt{x_i} | f_{H_{i-1}}\rangle & H_{i-1} & \sqrt{x_i}|l_{H_{i-1}}\rangle \\
0 & \sqrt{x_i}\langle l_{H_{i-1}}| & 0
\end{pmatrix}
\end{equation}
and $H_0 = 0$. Finally let $N_i = \prod_{j = 1}^i x_j$.

There exists a unitary $U$ such that
    \begin{equation}
        U^\dagger A^X U = H_{d} \oplus  \bigoplus_{i = 0}^{d-1} (\mathbb I_{(x_{i+1} - 1)N_d / N_{i+1} } \otimes H_{i}).
    \end{equation}
    In particular, the eigenvalue spectrum of $A^X$ contains many copies of the smaller tridiagonal matrices $H_i$ for $i = 0,1,\dots,d$.

\end{thm}
\begin{proof}
    This result can be acquired either by directly exploiting recursive structure in the characteristic polynomials of glued trees at various depths, or by an explicit construction of eigenstates. For the former approach, see Appendix \ref{app:chrpolyrec}. For the latter, see instead Appendix \ref{app:chrpolystat}.
\end{proof}
Theorem \ref{thm:fluxless} was well within the reach of \cite{childs}, and would certainly have been produced there if the spectra of glued trees were of any interest in and of themselves.
Although $A^{X}$ is a very large matrix in principle (for $p$-nary trees, $A^X$ is approximately $2 p^d$ by $2 p^d$ ), one can compute the full spectrum of $A^X$ by diagonalizing $d$ tridiagonal matrices that are no larger than $2 d + 1$ by $2 d + 1$. Hence, the spectrum of $A^X$ may be computed in time $t\sim d^3$.

In the case of the $p$-nary tree of depth $d$, the eigenvalues of the $H_i$ matrices are known to be $2 \sqrt{p} \cos( r \pi)$ where $r$ is a positive rational number less than one.
In cases where $\mathcal G(A^X)$ is not $p$-nary, a simple formula for the eigenvalues of the $H_i$ matrices is not known to us, but the characteristic polynomial of a tridiagonal matrix does satisfy a simple recurrence relation.
\begin{prop}\label{obs:avgdeg}
    The average vertex degree of a glued tree is less than $4$.
\end{prop}
\begin{proof}
    A tree with $|V|$ vertices can have at most $|V|-1$ edges since the trees do not contain loops, by definition. It follows immediately that the average vertex degree of a tree is less than two. Let $T = (E_t, V_t)$ be a tree and let $G = (E,V)$ be the graph constructed by making a copy of $T$, say $T'$ and identifying every vertex of degree one in $T$ with its copy in $T'$. This construction produces every glued tree as well as objects outside of the set of glued trees as we have defined them.
    The average vertex degree in $T$ is given by
    \begin{equation}
        \frac{1}{|V_t|}\sum_{v \in V_t} \deg(v) < 2.
    \end{equation}
    Let $V_t^1 \subset V_t$ be the set of vertices of degree one in $T$. Then,
    \begin{equation}
        \frac{1}{|V_t|} \sum_{v \not\in V_t^1} \deg(v) + \frac{|V_t^1|}{|V_t|} < 2.
    \end{equation}
    The average vertex degree in $G$ is
    \begin{equation}
        \frac{2}{2 |V_t| - |V_t^1|} \sum_{v \not\in V_t^1} \deg(v) + 2 \frac{|V_t^1|}{2|V_t| - |V_t^1|} < 4
    \end{equation}
    since $|V_t^1| \leq |V_t|$.
\end{proof}
To see that Proposition \ref{obs:avgdeg} is tight, observe that the $p$-shrub has average vertex degree $\frac{4 p}{p + 2} \sim 4$ when $p$ is large.
For $p$-nary trees of depth $d$, average vertex degree is
\begin{equation}
    4\frac{p(p^d - 1)}{p^d(1 + p) - 2} \lesssim 4 \frac{p}{p+1}
\end{equation}
as $d$ becomes large.

We now specify a particular embedding of glued trees that will be of value when we introduce complex weights (hopping phases) to the problem.
Let $G$ be a graph. Assign to every vertex in $G$ a point in $\mathbb R^2$, say $r(v)$. For every edge in $G$, define a smooth differentiable function $f_{(u,v)} : [0,1] \rightarrow \mathbb R^2$ such that
$f_{(u,v)}(0) = r(u)$ and $f_{(u,v)}(1) = r(v)$. If $f_e(t) = f_{e'}(t')$ implies that $e = e'$ and $t = t'$ or $t,t' \in \{0,1\}$, we say that the collection of functions $f_e$ together with the function $r$ is a \textit{planar embedding} of $G$. We consider a pair of planar embeddings to be equivalent if they are related by some smooth deformation of the functions that define the embedding.
The details of the maps that make up an embedding are unimportant to us.
We also do not distinguish between embeddings of a graph related by relabeling edges and vertices.
In particular, two planar embeddings of a graph are distinct only if their duals are not graph isomorphic.
\begin{defn}[Canonical embedding]\label{defn:canonical_embedding}
   Let $G$ be a glued tree of depth $d$. We say that a planar embedding of $G$ is a \textbf{canonical embedding} if the length of the cycle bounding the exterior face is $4d$.
\end{defn}
\begin{figure}
    \centering
    \includegraphics[width=0.6\linewidth]{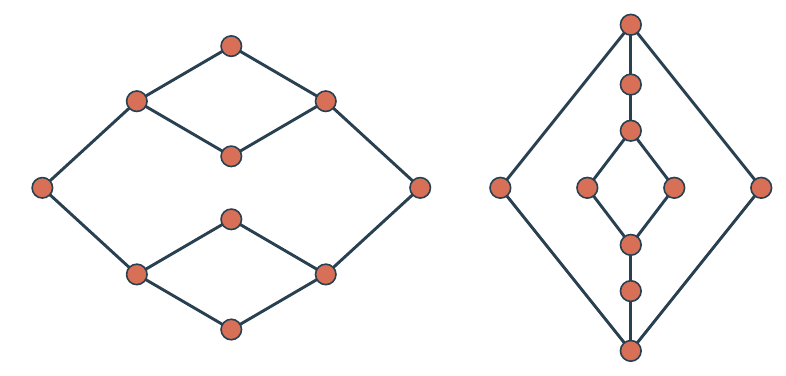}
    \caption{Two planar embeddings of the glued tree grown from $X = \{2,2\}$.
    The embedding on the left is canonical while the one on the right is not. The complex weighting produced by the demand that every interior face be pierced by the same flux is dependent upon the choice of embedding.
    }
    \label{fig:embedding}
\end{figure}
We will henceforth suppose that glued trees are canonically embedded. In particular, a \textit{plaquette} is one of the bounded faces in a canonical embedding of a glued tree.
Figure \ref{fig:embedding} compares the canonical embedding of the glued tree grown from $X = \{2,2\}$ with a noncanonical one, in which a different set of cycles bounds the plaquettes.

\begin{defn}[Weighted adjacency matrix]\label{defn:weightedadjacency}
   Let $G$ be a graph and $A = \mathcal A(G)$ be its adjacency matrix. Let $P$ be a (not necessarily real) $|V|$ by $|V|$ matrix.
   Define the matrix $A_P$ with entries
   \begin{equation}
       \langle u | A_P |v\rangle =  \langle u| P|v\rangle \langle u | A | v\rangle.
   \end{equation}
   $A_P$ is called a \textbf{weighted adjacency matrix} with weights $P$.
\end{defn}
Weighted adjacency matrices in general are great in number, and every example of a weighted adjacency matrix may be regarded as ``too special" to yield much mathematical interest as objects of study in the spectral graph theory community. Nonetheless, they do arise in any number of practical problems.
This is a sensible opinion since every finite square matrix can be recast as some manner of weighted adjacency matrix.
Weighted adjacency matrices are even more difficult to study in general when the entries of $P$ are allowed to be complex.

In keeping with our physical motivation of studying Hubbard models, it is natural to require that $P^\dagger = P$, since quantum Hamiltonians must be Hermitian.
Another natural restriction is to require that the entries of $P$ are pure phases. In part, this is motivated by an interest in studying interference effects on lattices, which give rise to interesting phenomenology.
To that effect, we also demand that $| \langle u | P | v\rangle|  = 1$.
It turns out that this set of restrictions leaves only $|E| - |V| + 1$ independent parameters needed to specify $P$ (provided that the graph in question is connected and planar).

What remains under consideration is Hermitian matrices $H = H^\dagger$ with entries of magnitude one. Such matrices are interpreted as a Hamiltonian describing the hopping of a boson (or indeed a fermion at the single-particle level) coupled to some magnetic or other gauge field. Mathematically, one might say that $H$ generates walks on some graph where particles are allowed to interfere like waves.
One of the conventional motivations for spectral graph theory is to understand properties of combinatorial
or probabilistic problems. We hope to convince the reader that interfering walks are a natural and sensible motivation for the study of the eigenvalues
of complex-weighted adjacency matrices.

It is no accident that $|E| - |V| + 1$ is also the number of internal faces in some planar embedding of a graph.
Suppose
\begin{equation}
    \langle u| P |v \rangle = e^{\mathrm i \phi_{uv}},
\end{equation}
with $\phi_{uv}$ real for all $u$ and $v$.
This is no loss of generality given the restrictions we have imposed.
$P = P^\dagger$ means that $\phi_{vu} = - \phi_{uv}$.
Let $U = |w \rangle \langle w |(e^{\mathrm i \Gamma}-1) + \mathbb I_{|V|}$. Clearly $U^\dagger U = \mathbb I_{|V|}$. Hence the spectrum of $U^\dagger H U $ is equal to the spectrum of $H$. However,
\begin{equation}
    \langle u | U^\dagger H U|v\rangle = \sum_{u,v} e^{\mathrm i \phi'_{uv}} \langle u | A | v\rangle
\end{equation}
with
\begin{equation}\label{eqn:phiredefn}
    \phi_{uv}' = \phi_{uv} + \Gamma (\delta_{wv} - \delta_{wu}).
\end{equation}
In this way, phases can be freely moved from edge to edge, but the product of phases around some loop is unchanged.
Thus, redefinitions of $\phi$ of the form (\ref{eqn:phiredefn}) leave the eigenvalue spectrum of a complex-weighted adjacency matrix unchanged. Such transformations are called \textit{gauge transformations} by the physics community. It follows straightforwardly, then, that the eigenvalue spectrum of a hopping Hamiltonian depends only upon the value of $\phi$ summed about a linearly independent\footnote{By ``linearly independent", we refer to linear independence in the vector space where addition of edges is carried out in $\mathbb Z_2$. } set of loops, of which there are $|E| - |V| + 1$ in number. $\phi$ acts effectively as a boundary map between two cells and one cells in a particular embedding of a graph, and it follows from the fact that the 2 sphere has no boundary that the phase around any particular loop is determined by the phases around the other loops.
This paper restricts its attention further to complex-weighted adjacency matrices for which the value of $\phi$ summed around each internal face of an embedding is the same.
By making this additional demand, we restrict our attention to adjacency matrices where the effects of interference are fully determined by a single number.
We make this discussion precise with the following definition.
\begin{defn}\label{defn:flux}
    Let $G$ be a canonically embedded glued tree. Let $\phi$ be an antisymmetric $|V|$ by $|V|$ matrix with elements $\langle u | \phi | v\rangle = \phi_{uv}$.
    Suppose that the sequence $l = \{v_1, v_2, \dots , v_n,v_1\}$ is a sequence of vertices traversed while tracing around a loop that bounds some interior face of $G$, which is unambiguous by virtue of our choice of embedding (See Figure \ref{fig:embedding}). Define the \textit{flux} through the face to be
    \begin{equation}
        \Phi_l = \phi_{v_1 v_2} + \phi_{v_2 v_3} + \dots + \phi_{v_n v_1}.
    \end{equation}
    If the flux through every internal face of $G$ is the same, we write the common value as $\Phi$, and we write
    \begin{equation}
        \mathcal A(G)(\Phi) = \sum_{u,v} e^{\mathrm i \phi_{uv}} |u \rangle \langle u | \mathcal A(G) |v\rangle \langle v|.
    \end{equation}
    $\mathcal A(G)(\Phi)$ is called a \textbf{canonical complex adjacency matrix} or CCAM for short.
\end{defn}
Importantly, without some generally applicable way of choosing an embedding, it would not be possible to refer to the faces of a glued tree without ambiguity. In lieu of an embedding choice, one must choose a spanning set of graph cycles (with orientation) and specify their fluxes individually.

To reiterate on our notation, if $G$ is a glued tree and  $A = \mathcal A(G)$ is its adjacency matrix, then $A(\Phi) = \mathcal A(G)(\Phi)$ is its CCAM.
The word ``flux" is meant to evoke the role of the magnetic vector potential in Aharonov--Bohm interference.
Although we only define CCAMs for canonically embedded glued trees, a similar construction is possible for any embedding of a graph on a surface with any genus.
Nonplanar graphs have a cyclomatic number that differs from the number of faces in an embedding.
We remark that the designation of a matrix as a CCAM for a graph is unique only up to gauge transformations, which we regard as unimportant.
We will suppress further discussion of various gauge transformations where possible.

We are not able to explicitly produce eigenvalue spectra for CCAMs in general. We are nevertheless able to produce the following:
\begin{thm}\label{thm:fluxaf}
    Let $X = \{x_1, x_2, \dots, x_d\}$ be a sequence of positive integers with all $x_i > 1$ (for simplicity).
    $A^X(\Phi)$ is then a CCAM. With $\Phi_0 = \frac{2 \pi}{x_1}$, there exists a unitary $U$ such that
    \begin{equation}
    U^\dagger A^X(\Phi_0)U   = \mathbb I_2 \otimes \left(F_{d} \oplus \bigoplus_{i = 2}^{d} (\mathbb I _{(x_i - 1) N_d/N_{i}} \otimes F_{i-1} ) \right) \oplus (\mathbb{I}_{(x_1 -2)N_d} \otimes F_0)
    \end{equation}
    with
    \begin{equation}
        F_i = \sum_{m = 1}^{i} \sqrt{x_{m}}\left( |m-1 \rangle \langle m| + |m \rangle \langle m-1 | \right)
    \end{equation}
    for $i \geq 1$, $F_0 = 0$, and $N_i = \prod_{j=2}^{i} x_j$.

\end{thm}
\begin{proof}
    Once again, this result may be acquired by either solving a recurrence relation or by explicitly producing eigenstates. See Appendices \ref{app:chrpolyrec}, \ref{app:chrpolyafb}, and \ref{app:chrpolystat}.
\end{proof}
$\Phi_0 = \frac{2 \pi}{x_1}$ is very special.
In particular, in glued trees with nonconstant $X$, it is not possible to produce explicit eigenvalues of $A^X(\Phi_0)$ in general, but it is possible to produce the eigenvalues of $A^X(\Phi_0)$ by instead diagonalizing $d$ tridiagonal matrices that are at most $d$ by $d$. Theorem \ref{thm:fluxaf} thus provides instructions for calculating the eigenvalues of very large glued trees and generalizes Theorem \ref{thm:fluxless}. Theorems \ref{thm:fluxless} and \ref{thm:fluxaf} are the means by which we have produced the spectra displayed in Figure \ref{fig:eigvals} which would have otherwise required diagonalizing matrices of size at least $2^{2000}$ by $2^{2000}$.
We do not expect that solutions of this nature can be found at generic values of $\Phi$ on the grounds that to do so would be to solve the eigenvalue problem for $A^X(\Phi)$ at any value of $\Phi$. However, we also do not give a formal argument that this is impossible. It would be interesting if some such solution were found.

\section{Glued tree lattices}\label{sec:lattices}
We now move on from discussion of finite glued trees and onto infinite graphs constructed by using many glued trees.
The simplest of these is, unsurprisingly, a one-dimensional chain of glued trees.
The earlier parts of this section are meant to be a rigorous discussion, where we prove our main theorem.
The latter half is instead a lattice exotica, giving many examples of lattices with all flat bands and only compact localized states, of as great a variety as we can manage to fit in a paper of reasonable length.

\subsection{Compact localized states from Aharonov--Bohm interference}\label{sec:abstates}
\subsubsection{A useful choice of gauge}\label{sec:gauge}
    As much as possible, we suppress discussion of particular gauge choices since they are unimportant to both eigenvalue spectra and to other physically important properties of the various models we consider.
    However, all explicit calculations must be carried out in \textit{some} gauge. We present here a gauge that is particularly amenable to the calculations carried out in Appendices \ref{app:chrpolyrec} and \ref{app:chrpolystat} as well as for the remainder of section \ref{sec:lattices}.

    \begin{defn}\label{defn:canonicalphases}
    Let $\Phi$ and $X = \{x_1, x_2, x_3, \dots, x_d\}$  be a real number and a sequence of positive integers respectively. Let $\{\omega_i\}_{i = 1}^d$ be the sequence that satisfies
    \begin{equation}
        \omega_i =\frac{\Phi}{4} (x_i - 1) \prod_{j = 1}^{i -1} x_j
    \end{equation}
    Furthermore, let
    \begin{equation}
        |\omega_j\rangle = \sum_{i = 1}^{x_{j}} \text{exp}\left[ \mathrm i \left(1 - \frac{2(i-1)}{x_{j}-1}\right) \omega_j\right] | i \rangle.
    \end{equation}
    $|\omega_i\rangle$ is called a \textbf{canonical phase vector}.
\end{defn}
\begin{lem}\label{lem:zero}
    Let $\Phi$ and $X = \{x_1, x_2, x_3, \dots, x_d\}$  be a real number and a sequence of positive integers respectively. Let $|\omega_i\rangle$ be canonical phase vectors for $i = 1, 2, \dots, d$.
    \begin{equation}
        \langle \omega_i | \left(|\omega_i\rangle \right)^* = 0
    \end{equation}
    if and only if there exists an integer $N$ not divisible by $x_i > 1$ such that $\Phi = N\frac{2 \pi}{\prod_{j = 1}^{i}x_j} $.
\end{lem}
\begin{proof}
First note that
\begin{equation}\label{eqn:geom}
    \sum_{l = 1}^m e^{\mathrm i \theta l} = 0
\end{equation}
if and only if there exists an integer $N$ not divisible by $m$ such that $\theta = 2 \pi N / m$.
This is a fact with a nice geometrical interpretation: the center of mass of a collection of equivalent particles evenly spaced along a circular arc about the origin lies precisely at the origin.
    With (\ref{eqn:geom}) in mind, observe that
    \begin{equation}\label{eqn:vanish}
    \begin{aligned}
        \langle \omega_i | \left(|\omega_i \rangle \right) ^* &= \sum_{j = 1}^{x_i}\text{exp}\left[2 \mathrm i \left(1 - \frac{2(j-1)}{x_i - 1}\right)\omega_i \right]  = e^{\mathrm i (2 \omega_i + \theta) } \sum_{j = 1}^{x_i} \text{exp}\left[ - \mathrm i j \theta \right]
        \end{aligned}
    \end{equation}
    with
    \begin{equation}
        \theta = \frac{4}{x_i - 1}\omega_i = \Phi \prod_{l = 1}^{i - 1} x_l.
    \end{equation}
    Then, (\ref{eqn:geom}) implies that (\ref{eqn:vanish}) is equal to zero exactly when there is an integer $N$  not divisible by $x_i$ such that
    \begin{equation}
        \Phi \prod_{l = 1}^{i-1}x_l = \frac{2 \pi}{x_i} N.
    \end{equation}
    Hence, the lemma holds.
\end{proof}
\begin{defn}[Canonical gauge]\label{defn:canonicalgauge}
    Let $X = \{x_1, x_2 , \dots x_d\} $ be a sequence of positive integers such that $x_d$ is greater than one.
    Let $Y_1, Y_2, \dots, Y_d$ satisfy the partitioned matrix recurrence relation
    \begin{equation}\label{eqn:cgrowth}
        \begin{aligned}
            Y_{i} &= \begin{pmatrix}
                0 & \langle \omega_i  |^* \otimes \langle f_{Y_{i-1}} |   & 0 \\
                |\omega_i \rangle^* \otimes | f_{Y_{i-1}}\rangle  & \mathbb I_{x_i} \otimes Y_{i-1} & |\omega_i \rangle \otimes | l_{Y_{i-1}}\rangle \\
                0 & \langle \omega_i | \otimes \langle l_{Y_{i-1}} | & 0
            \end{pmatrix} \\
            Y_1 &= \begin{pmatrix}
                0 &\langle \omega_1|^* & 0 \\
                |\omega_1\rangle^* & 0 & |\omega_1\rangle \\
                0 & \langle \omega_1| & 0
            \end{pmatrix}
        \end{aligned}
    \end{equation}
    where $\otimes$ is the Kronecker product.
    We say that $Y_d$ gives the \textbf{canonical gauge} of $A^X$, and we write $A^X(\Phi) := Y_d$.
\end{defn}
It remains to show that Definition \ref{defn:canonicalgauge} constitutes a choice of gauge with the restrictions demanded above (namely that every loop in a canonical embedding of $\mathcal G(A^X)$ winds the same phase). To see that this is the case, observe that
\begin{equation}
     \left( 2 \langle j | \omega_i \rangle - 2 \langle j + 1| \omega_i \rangle  + 4 \sum_{l = 1}^{i-1} \omega_l  \right) = \Phi.
\end{equation}
Hence, the phase wound around every internal face in a canonical embedding of $\mathcal G(A^X)$ is precisely $\Phi$.
The recursive structure provided by (\ref{eqn:cgrowth}) is extremely useful for proving various facts about glued trees and their spectra.

\subsubsection{Glued trees cannot be crossed at some values of $\Phi$}\label{sec:confine}
In this section, we will need to make use of various standard objects. Firstly, we use $C(M;\lambda) := \det(M - \lambda)$ to denote the \textit{characteristic polynomial} of a matrix, and $R(M;\lambda) := (M - \lambda)^{-1}$ to denote the \textit{resolvent} of a matrix. The \textit{adjugate} of a matrix is the transpose of its cofactor matrix. If  $M$ is invertible, then $M^{\mathrm A} = \det(M) M^{-1}$. When $\lambda$ is in the resolvent set of $M$, we have that $(M - \lambda)^{\mathrm A} = C(M;\lambda) R(M;\lambda)$.
\begin{defn}[Flat values]\label{defn:flat}
    Let $X = \{x_1, x_2, \dots, x_d\}$ be a set of positive integers. Define
    \begin{equation}
        F^X = \left\{ \frac{2 \pi}{\prod_{i = 1}^d x_i} z \,:\, z = 1,2,\dots, \prod_{i =1}^d x_i - 1\right\}.
    \end{equation}
    We say that $F^X$ is the set of \textbf{flat values} of $X$.
\end{defn}
The flat values associated with a given sequence of integers $X$ depend only upon the product of the values in $X$.
Explicitly,
let $X = \{x_1, x_2, \dots, x_d\}$ be a set of positive integers with \emph{all} $x_i > 1$. $X$ encodes an ordered factorization of $M = \prod_{i = 1}^d x_i$. Let $N(M)$ be the number of sets of integers greater than one whose product is equal to $M$. It is well known that the Dirichlet generating function of ordered integer factorizations \cite{CHOR2000123} is given by
\begin{equation}\label{eqn:generating}
    \frac{1}{2 - \zeta(s)} = \sum_{M = 2}^\infty \frac{N(M)}{M^s},
\end{equation}
where $\zeta(s)$ is the Riemann zeta function.
Hence, the number of glued trees with flat values $F^X$ at fixed $\prod_{x \in X} x$ is \emph{also} generated by $\frac{1}{2 - \zeta(s)}$.

The following lemmas and theorems should reveal the intuition for the name that we have chosen for this special subset of $2 \pi \mathbb Q$.
\begin{lem}\label{lem:travel}
    Let $X = \{x_1, x_2, \dots, x_d\}$ be a set of positive integers with $x_d > 1$ and $\Phi$ real. Suppose $A^X(\Phi)$ is constructed in the canonical gauge according to the procedure given in Definition \ref{defn:canonicalgauge}.
    If $\Phi \in F^X$, then for any $\lambda$ outside the spectrum of $A^X(\Phi)$,
    \begin{equation}\label{eqn:obj}
        0 = \langle f_{A^X(\Phi)} | R\left(A^X(\Phi);\lambda\right)| l_{A^X(\Phi)}\rangle.
    \end{equation}
\end{lem}
\begin{proof}
    In practice, it is easier to exploit the recursive structure provided by Definition \ref{defn:canonicalgauge} using the adjugate than using the resolvent. Consider (\ref{eqn:cgrowth}), which provides
    \begin{equation}
        A^X(\Phi) = Y_d
        = \begin{pmatrix}
                0 & \langle \omega_d  |^* \otimes \langle f_{Y_{d-1}} |   & 0 \\
                |\omega_d \rangle^* \otimes | f_{Y_{d-1}}\rangle  & \mathbb I_{x_d} \otimes Y_{d-1} & |\omega_d \rangle \otimes | l_{Y_{d-1}}\rangle \\
                0 & \langle \omega_d | \otimes \langle l_{Y_{d-1}} | & 0
            \end{pmatrix}.
    \end{equation}
    By using the matrix determinant lemma, it follows that
    \begin{equation}
        \langle f_{Y_d} | (Y_d - \lambda)^{\mathrm A} | l_{Y_d}\rangle = \sigma C(Y_{d-1};\lambda)^{x_d - 1}(\langle \omega_d | )^* |\omega_d\rangle  \langle l_{Y_{d-1}}| (Y_{d-1} - \lambda)^{\mathrm A} | f_{Y_{d-1}}\rangle
    \end{equation}
    where $\sigma$ is a sign that is unimportant to the present argument. See Appendix \ref{app:chrpolyrec} and Lemma \ref{lemma:huge} in particular for careful accounting of this recurrence.
    The recurrence terminates at the shrub $Y_1$, where $\langle f_{Y_1} | (Y_1 - \lambda)^{\mathrm A} | l_{Y_1}\rangle = \pm(-\lambda)^{x_1 - 1} (\langle \omega_1|)^* |\omega_1\rangle$.
    Repeated application of this recurrence shows that
    \begin{equation}\label{eqn:objective}
        \langle f_{A^X(\Phi)} | \left(A^X(\Phi)-\lambda\right)^{\mathrm A}| l_{A^X(\Phi)}\rangle=0
    \end{equation}
    whenever
    \begin{equation}\label{eqn:intermedkill}
        \prod_{i = 1}^d (\langle \omega_i|)^* |\omega_i\rangle = 0.
    \end{equation}
    By Lemma \ref{lem:zero}, (\ref{eqn:intermedkill}) holds whenever  $\Phi \in F^X$.
    Indeed, write $\Phi = 2 \pi z / \prod_{j = 1}^d x_j$ with $1 \leq z \leq \prod_{j = 1}^d x_j - 1$. Since $\prod_{j = 1}^d x_j$ does not divide $z$, there is a largest index $i$ such that $\prod_{j = i}^d x_j$ does not divide $z$, and by maximality $N := z / \prod_{j = i + 1}^d x_j$ is an integer (the empty product being one). Then $x_i$ does not divide $N$, and $\Phi = 2 \pi N / \prod_{j = 1}^{i} x_j$, so Lemma \ref{lem:zero} gives $(\langle \omega_i |)^* |\omega_i\rangle = 0$.
    Now, provided that $\lambda$ is not in the spectrum of $A^X(\Phi)$, $R(A^X(\Phi);\lambda)$ exists and satisfies
    \begin{equation}
	  R(A^X(\Phi);\lambda) = \frac{1}{C(A^X(\Phi);\lambda)} (A^X(\Phi) - \lambda)^{\mathrm A}.
    \end{equation}
    As $C(A^X(\Phi);\lambda)$ is nonvanishing when $\lambda$  is outside the spectrum of $A^X(\Phi)$,
    (\ref{eqn:objective}) gives
    \begin{equation}
    \langle f_{A^X(\Phi)} |  R(A^X(\Phi);\lambda) | l_{A^X(\Phi)}\rangle=
    \frac{1}{C(A^X(\Phi);\lambda)}\langle f_{A^X(\Phi)} | \left(A^X(\Phi)-\lambda\right)^{\mathrm A}| l_{A^X(\Phi)}\rangle=0.
    \end{equation}
    Thus, the proof is complete.
\end{proof}
The technique of using the recursive structure of the canonical gauge to calculate matrix elements of the resolvent of $A^X(\Phi)$ is the means by which one is able to produce formulae for the spectrum of glued trees at $\Phi = 0$, and at $\Phi = 2 \pi/x_1$. We remind the reader once again that a thorough set of calculations along these lines can be found in Appendices \ref{app:chrpolyrec} and \ref{app:chrpolyafb}.

Lemma \ref{lem:travel} can be restated in a more intuitive way:
\begin{cor}[Aharonov--Bohm interference]\label{cor:nocross}
    Adopt the hypothesis of Lemma \ref{lem:travel} and let $\Phi$ be one of the special points ordained therein.
    Then,
    \begin{equation}
        \langle l_{A^X(\Phi)} | (A^X(\Phi))^m | f_{A^X(\Phi)}\rangle =
        \langle l_{A^X(\Phi)} | \text{exp}\left[\mathrm i t A^X(\Phi) \right]| f_{A^X(\Phi)}\rangle = 0
    \end{equation}
    for any positive integer $m$ and every real $t$.
\end{cor}
\begin{proof}
    Let $\sigma$ be the set of eigenvalues of $A^X(\Phi)$, and, for $\mu \in \sigma$, let  $P_\mu$ be the projector onto the kernel of $(A^X(\Phi) - \mu)$.
    Then, the spectral decomposition of $A^X(\Phi)$ is given by $A^X(\Phi) = \sum_{\mu \in \sigma} \mu P_\mu$.
    For $\lambda$ not in $\sigma$,
    \begin{equation}
        \langle f_A | R(A^X(\Phi);\lambda) | l_A \rangle =
        \sum_{\mu \in \sigma} \frac{\langle f_{A^X(\Phi)} | P_\mu | l_{A^X(\Phi)} \rangle}{\mu - \lambda},
    \end{equation}
    and by Lemma \ref{lem:travel} the left hand side vanishes identically.
    Pick some $\mu_0$ in $\sigma$, and
    \begin{equation}
	  0 = \lim_{\lambda \rightarrow \mu_0} (\lambda - \mu_0) \langle f_A | R(A^X(\Phi);\lambda) | l_A \rangle  =
			\lim_{\lambda \rightarrow \mu_0} (\lambda - \mu_0)
			\sum_{\mu \in \sigma} \frac{\langle f_{A^X(\Phi)} | P_\mu | l_{A^X(\Phi)} \rangle}{\mu - \lambda} =
			 \langle f_{A^X(\Phi)} | P_{\mu_0} | l_{A^X(\Phi)} \rangle.
    \end{equation}
    Hence, $\langle f_{A^X(\Phi)}|P_\mu | l_{A^X(\Phi)}\rangle = \langle l_{A^X(\Phi)}|P_\mu | f_{A^X(\Phi)}\rangle = 0$ for every $\mu$ in $\sigma$. The claims of the Corollary are then immediate.
\end{proof}
Effectively, we have proven that $A^X(\Phi)$ is impossible to cross when $\Phi \in F^X$. We will use this result to prove our main theorem.

Now, we state and prove our main result.
\begin{thm}\label{thm:main}
    Let $G = (E,V)$ be a graph, and
   let $A_P$ (which we regard as a linear map on $\ell^2(V)$) be a Hermitian complex-weighted adjacency matrix on $G$. Further suppose that the maximum vertex degree $D$ of $G$ is finite.
   Suppose there exists a positive integer $N$ such that
   \begin{equation}
       \langle a | A_P^m | b\rangle = 0,
   \end{equation}
   for every positive integer $m$ whenever $\text{dist}(a,b) \geq N$, where $\text{dist}(a,b)$ is the graph distance between $a$ and $b$. Then, $\ell^2(V)$ admits an orthonormal basis of compact localized eigenstates of $A_P$, and the eigenvalues of $A_P$ are dense in its spectrum.
\end{thm}
\begin{proof}
    Write
    \begin{equation}
	B_r(b):= \{u \in V\,:\, \text{dist}(u,b) \leq r\},
    \end{equation}
    and observe that the finiteness of $D$ implies that $|B_r(b)| \leq \sum_{k = 0}^{r} D^k < \infty$.
	Furthermore, each row and column of $A_P$ has at most $D$ nonzero entries, all of unit norm, so the Schur test gives $\|A_P\| \leq D$. A bounded Hermitian operator defined on all of $\ell^2(V)$ is self-adjoint, so the orthogonal complement of any closed $A_P$-invariant subspace is $A_P$-invariant, and the orthogonal projector onto such a subspace commutes with $A_P$.

	Enumerate the vertices of $G$ as $V = \{v_1, v_2, \dots\}$.
	Set $Q_0 = 0$ and define inductively, for $k \geq 1$,
	\begin{equation}
	\begin{aligned}
	    \mathcal H_1 &= \text{span}\left\{A_P^m |v_1\rangle \,:\, m = 0,1,\dots\right\} \\
		\mathcal H_k &= \text{span}\left\{ A_P^m (1 - Q_{k-1}) |v_k\rangle \,:\, m = 0, 1, \dots \right\},
	\end{aligned}
	\end{equation}
	where $Q_{k-1}$ is the orthogonal projector onto $\mathcal H_1 \oplus \dots \oplus \mathcal H_{k-1}$ and $A_P^0|\psi\rangle = |\psi\rangle$ is understood.

	We proceed by induction.
	First, we observe that vectors in  $\mathcal H_1$ have support only on vertices in $B_{N-1}(v_1)$, which are finitely many in number, and thus $\mathcal H_1$ is finite-dimensional and $A_P$-invariant by construction, and likewise $\mathcal H_1^\perp$ is $A_P$-invariant by the self-adjointness of $A_P$.
	Immediately, it admits an orthonormal basis of finitely supported eigenstates of $A_P$.

	Now, suppose that $\mathcal H_1, \dots, \mathcal H_{k-1}$ are finite-dimensional, $A_P$-invariant, pairwise orthogonal, and admit bases of finitely supported eigenvectors of $A_P$.
	Then, by the self-adjointness of $A_P$, $[Q_{k-1},A_P]=0$, so
	\begin{equation}
	A_P^m(1 - Q_{k-1})|v_k\rangle = (1 - Q_{k-1}) A_P^m |v_k\rangle.
	\end{equation}
     As before, $A_P^m|v_k\rangle$ is supported only on vertices in $B_{N-1}(v_k)$, which is finite, and by the inductive hypothesis $Q_{k-1}$ has finite rank and nonzero entries on only finitely many vertices.
	Thus, $\mathcal H_k$ is $A_P$-invariant and orthogonal to $\mathcal H_1 \oplus \dots \oplus \mathcal H_{k-1}$ by construction and the self-adjointness of $A_P$.
	Hence $\mathcal H_k$ is finite-dimensional and admits a basis of finitely supported eigenstates of $A_P$.
	Thus, the induction holds.

	Since $|v_k\rangle \in \mathcal H_1 \oplus \dots \oplus \mathcal H_k$ for every $k$, $\mathcal H:= \bigoplus_{k=1}^\infty \mathcal H_k$ is dense in $\ell^2(V)$, and the $\mathcal H_k$ are pairwise orthogonal.
	Hence the union over $k$ of the orthonormal eigenbases of the $\mathcal H_k$ is the required orthonormal basis of $\ell^2(V)$. Then, as a bounded operator, $A_P$ is unitarily equivalent to its diagonalization by the prescribed orthonormal basis, yielding a diagonal operator with entries equal to eigenvalues of $A_P$. Finally, the spectrum of a diagonal operator is the closure of its diagonal entries, and the spectrum of an operator is preserved by a unitary transformation, so the eigenvalues of $A_P$ are dense in its spectrum.
\end{proof}
The assumption of bounded degree is used only for the boundedness of $A_P$ and can be relaxed to the assumption that balls about any given vertex are finite in exchange for more care regarding the domain on which $A_P$ is self-adjoint.

Theorem \ref{thm:main} can be stated rather simply: If all particles are confined to a neighborhood by a Hamiltonian, then  all eigenstates can be written as a sum of other, local eigenstates.
Proving that eigenstates can be chosen to have finite support is actually stronger than the  condition of having flat bands, since a notion of band structure requires, at a minimum, a translation symmetry, while the property of having finite support is well defined regardless of symmetry.
The values of $\Phi$ for which all eigenstates are compact localized are precisely the flat points, which is the reason for the nomenclature.
For Euclidean lattices, all eigenstates are compact localized if and only if all bands of the Bloch Hamiltonian are flat \cite{email1,email2,email3}. Since we are interested in non-Euclidean lattices, however, we prefer to show that eigenstates are compact localized directly.
Furthermore, in one-dimensional models, the condition of having only flat bands is equivalent to the statement that eigenstate projectors are also localized \cite{email1}.

We now proceed to define infinitely many lattices that satisfy Theorem \ref{thm:main} and some that have only compact localized states despite failing to satisfy Lemma \ref{lem:travel}.
\subsection{Infinite graphs with compact localized states}\label{sec:examples}
\subsubsection{Replacing the edges in an infinite graph with glued trees}
\begin{cor}[of Lemma \ref{lem:travel} and Theorem \ref{thm:main}]\label{cor:main}
    Let $G = (E,V)$ be a (countable but not necessarily finite) graph of finite maximum degree, and $E' \subset E$ be a set of its edges.
    We now regard $\mathcal A(G)$ as a linear map on $\ell^2(V)$.
    Choose a real number $\Phi_e$, for each edge $e \in E'$, and a sequence of integers $X^e = \{x_1^e, x_2^e, \dots x_{d_e}^e\}$ such that $x_{d_e}^e > 1$ for all $e$, with only finitely many distinct sequences $X^e$ occurring.
    Let $U_e$ be the set of vertices of the glued tree grown from $X^e$ other than its two roots, with $U_{e} \cap U_{e'} = \emptyset$ if $e \neq e'$ and $U_e \cap V = \emptyset$ for all $e$, and regard $A_e := A^{X^e}(\Phi_e)$ as a linear map on $\ell^2(U_e \cup \{u_e, v_e\})$ in which $|f_{A_e}\rangle = |u_e\rangle$ and $|l_{A_e}\rangle = |v_e\rangle$, extended by zero to $\ell^2\left(V \cup \bigcup_{e \in E'} U_e\right)$.
    Define the linear map $\mathcal R(G,E')$ on $\ell^2\left(V \cup \bigcup_{e \in E'} U_e\right)$ by
    \begin{equation}
    \mathcal R(G, E')  = \mathcal A(G) - \sum_{e \in E'} \left(
            |u_e \rangle \langle v_e| + |v_e \rangle \langle u_e | \right) + \sum_{e \in E'} A_e.
    \end{equation}
    If there is an integer $S$ such that every connected component of the  graph $(E \setminus E', V)$ has at most $S$ vertices,
    and $\Phi_e \in F^{X^e}$ for each $e \in E'$, then $\ell^2(V \cup \bigcup_{e \in E'} U_e)$ admits an orthonormal basis of compact localized eigenstates of $\mathcal R(G,E')$, and the eigenvalues of $\mathcal R(G,E')$ are dense in its spectrum.
\end{cor}
\begin{proof}
    Let $G'$ be the graph on which $\mathcal R(G,E')$ is a complex-weighted adjacency matrix, which is explicitly $G$ with each $e \in E'$ replaced by the glued tree grown from $X^e$.
    Since $G$ has finite maximum degree and only finitely many distinct trees occur, $G'$ also has finite maximum degree.
    Thus, to check that Theorem \ref{thm:main} applies, we must verify that there exists an integer $N$ such that $\langle a | \mathcal R(G,E')^m | b \rangle = 0$ for every $m$ whenever $\text{dist}(a,b) \geq N$.

    To that end, recall that $\langle a | \mathcal R(G,E')^m | b \rangle$ can be expressed as a weighted sum of walks of length $m$ from $b$ to $a$, where the weight of a walk through vertices $b = w_0, w_1, \dots, w_{m} = a$ is given by $\prod_i \langle w_{i+1}| \mathcal R(G,E')| w_i \rangle$.
    Corollary \ref{cor:nocross} can then be interpreted as meaning that this complex-weighted sum of walks between the roots of a glued tree vanishes, regardless of the value of $m$.

    Since the connected components of $(E \setminus E',V)$ are all at most $S$ in size, and the sets $U_e$ are uniformly bounded, walks between vertices in $G'$ separated by $N:=2 \sup_{e \in E'} |U_e| + S$ distance or more must traverse across some glued tree.
    In particular, the two roots of a glued tree are never in the same connected component of $(E\setminus E',V)$, which have diameter at most $S - 1$, and glued trees have diameter equal to twice their depth.
    Suppose $a$ and $b$ are sufficiently distant in $G'$ that every walk from $a$ to $b$ in $G'$ must traverse across some glued tree, and observe
    \begin{equation}\label{eqn:decompsum}
	   \langle a | \mathcal R(G,E')^m | b\rangle = \sum_{ w_1,w_2, \dots, w_{m-1}}  \prod_{i = 0}^{m-1} \langle w_{i+1}| \mathcal R(G,E')|w_i \rangle
    \end{equation}
    where again $w_0 = b$  and $w_m = a$ is understood. Suppose a term in the sum (\ref{eqn:decompsum}) is nonzero. Then, there exists an edge $e = (u,v)$ in $E'$ and integers $j< k$ without loss of generality such that $\{w_{j-1}, w_{k+1}\}\cap U_e = \emptyset$,  $w_j = u$, $w_k = v$, and $w_i \in U_e \cup \{u,v\}$ for $j < i < k$.
    Take the smallest such $j$, where we also understand that $j = 0$ and $k = m$ automatically satisfy the endpoint condition.
    These conditions specify the subwalk of maximal length that is nonetheless within the glued tree inserted into $G'$ as a replacement for the edge $e$ in $G$.
    Rearranging factors, we find
    \begin{equation}
    \prod_{i = 0}^{m-1} \langle w_{i+1}| \mathcal R(G,E') | w_i\rangle =
    \prod_{i = 0}^{j-1} \langle w_{i+1}| \mathcal R(G,E') | w_i\rangle
    \prod_{i = j}^{k-1} \langle w_{i+1}| \mathcal R(G,E') | w_i\rangle
    \prod_{i = k}^{m-1} \langle w_{i+1}| \mathcal R(G,E') | w_i\rangle.
    \end{equation}
    Finally, Corollary \ref{cor:nocross} provides
    \begin{equation}
    \sum_{w_{j+1}, w_{j+2}, \dots, w_{k-1} \in U_e \cup \{u,v\}}
    \prod_{i = j}^{k-1} \langle w_{i+1}| \mathcal R(G,E') | w_i\rangle =  \langle l_{A_e}| A_e^{k - j} | f_{A_e}\rangle = 0.
    \end{equation}
    Since this factoring of each nonzero term in (\ref{eqn:decompsum}) and the subsequent vanishing of the sum is a canonical choice, and every term in the sum is covered exactly once, we have established that $\langle a | \mathcal R(G,E')^m|b\rangle = 0 $ for any sufficiently distant $a$ and $b$.
    Hence, Theorem \ref{thm:main} applies.
\end{proof}
Corollary \ref{cor:main} can be summarized in the following way:
\begin{enumerate}
    \item Construct some crystal lattice $G$.
    \item Replace the chosen subset of edges in the crystal with complex-weighted glued trees at flat points, such that every sufficiently long self-avoiding path must traverse across some complex-weighted glued tree.
    \item The resulting complex-weighted adjacency matrix (or quantum Hamiltonian) has perfectly localized eigenstates because the glued trees are impossible to cross, by Lemma \ref{lem:travel}, which permits the use of Theorem \ref{thm:main} since every sufficiently long path away from any vertex must traverse across a glued tree.
    \item To see that such operators have countably many distinct eigenvalues, simply compute the spectrum of $\mathcal R$ by diagonalizing $\mathcal R$ in the space of states reachable from $v \in V$ by repeated application of $\mathcal R$. There is a finite set of eigenvalues acquired in this way for every element of $V$, which is assumed to be countable.
\end{enumerate}

In what follows, we will almost entirely restrict our attention to constructions with some number of translation symmetries. This choice is a pragmatic one, since Bloch Hamiltonians are much more easily analyzed than generic operators on $\ell^2(\mathbb Z)$.
Furthermore, the explicit example we provide is a quasi-one-dimensional model, but we emphasize that Corollary \ref{cor:main} applies to lattices in any number of spatial dimensions, and even graphs that are aperiodic.
Over the course of the paper, we have imposed the requirement that glued trees grown from positive integer sets must terminate in a number greater than one.
We see now why this demand that $x_d > 1$ is no loss of generality since $x_d = 1$ can be equivalently recast as a redefinition of $E'$ in $\mathcal R(G,E')$.

\begin{figure}
    \centering
    \begin{subfigure}[c]{0.49\linewidth}
    \centering
    \includegraphics{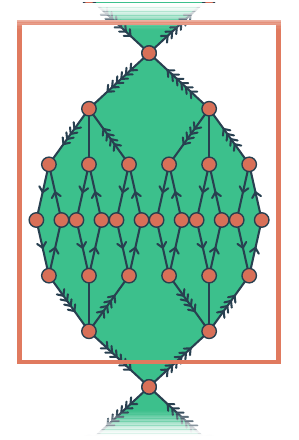}
    \end{subfigure}
    \begin{subfigure}[c]{0.50\linewidth}
    \centering
    \includegraphics{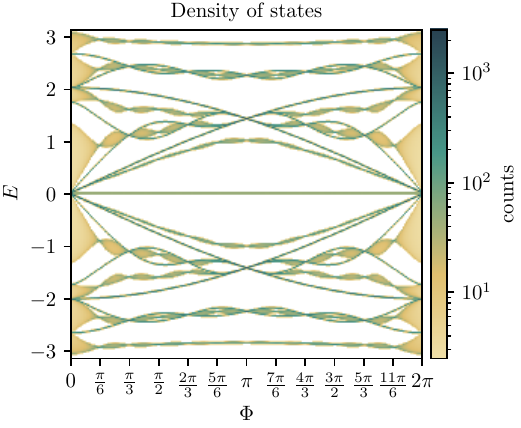}

    \end{subfigure}
    \caption{Above we have pictured a lattice that we have proven to have only compact localized eigenstates at flat points, and a plot showing that bands at flat points are themselves flat. \textbf{Left: } One unit cell (orange box) of the one-dimensional chain of glued trees grown from $X = \{2,3,2\}$ in the canonical gauge, together with parts of the neighboring cells. An arrow pointing from vertex $i$ to vertex $j$ denotes that $\langle i | H | j \rangle = e^{\mathrm i \frac{\Phi}{4}}$.
    Likewise four and six arrows between a pair of vertices denote $\langle i| H | j \rangle = e^{\mathrm i\Phi}$ and $\langle i | H | j \rangle = e^{\frac{3}{2}\mathrm i \Phi}$ respectively; edges without arrows carry no phase.
    \textbf{Right: } Above plotted is the density of states of the Bloch Hamiltonian of the $X = \{2,3,2\}$ one-dimensional chain of glued trees. Flat points $\Phi$ in  $F^{\{2,3,2\}}$ are labeled on the horizontal axis, where we see the total bandwidth vanish, indicating that all bands are dispersion free. }
    \label{fig:1d}
\end{figure}

Construct a graph with vertices given by the integers and edges connecting vertices separated by a distance (in $\mathbb Z$) equal to one. Replace every edge with a glued tree grown by $X = \{x_1, x_2, \dots, x_d\}$. Further demand that all $\Phi_e$ values are equal to $\Phi$.
If $\Phi \in F^X$, then the system at hand has compact localized states and only flat bands. The simplest example of one such lattice is given by the rhombi-chain lattice discussed in Section \ref{sec:warmup}.
Another example is given in Figure \ref{fig:1d}.

While the flatness of bands and localization of eigenstates is implied by Theorem \ref{thm:main}, lattices of this form can be shown to have flat bands without considering eigenstates at all. In particular, it is possible to show that this family of lattices has only flat bands by manipulating the characteristic polynomial of the Bloch Hamiltonian alone.
In practice, we are interested in properties of eigenstates, but it is nonetheless interesting that a discussion of eigenstates turns out to be unnecessary to show that bands are flat in translation-invariant chains. See Appendix \ref{app:chrpolyafb}.

We also remark that it is possible to derive the eigenvalue spectrum of $p$-nary glued tree lattices at $\Phi = 0$, and at $\Phi = \frac{2\pi}{p}$. Of course $0$ is not a flat point, so there is nontrivial dispersion but we regard any quantitative statements about the particular eigenvalues of these lattices as interesting. In particular, the aforementioned results are mild extensions of Theorems \ref{thm:fluxaf} and \ref{thm:fluxless}, which we again relegate to the appendices.
Finally, notice that  (\ref{eqn:generating}) implies that, for a given $G$ and $E$, the number of sets $X$ of integers greater than one for which $\mathcal R(G,E)$ has its minimum flat point at $\frac{2\pi}{M}$ is a sequence (of $M$) generated by the left hand side of (\ref{eqn:generating}).

Consider the lattice depicted in the left panel of Figure \ref{fig:1d}. The glued trees in question are grown from the sequence $X = \{2,3,2\}$.
The Bloch Hamiltonian in the canonical gauge is given by
\begin{equation}
    H =
    \begin{gathered}
    \includegraphics[width=.9\linewidth]{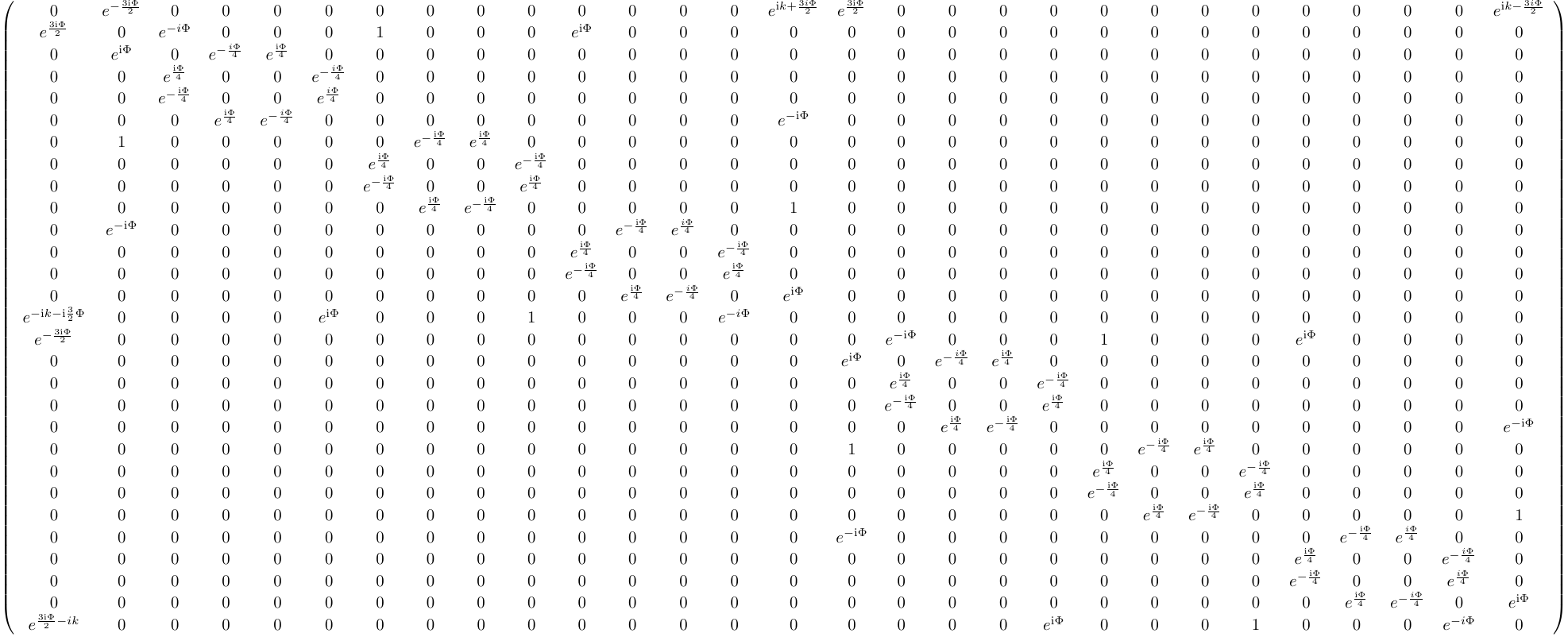}
    \end{gathered}
\end{equation}
Verifying that $H$ has only flat bands whenever $\Phi \in F^{\{2,3,2\}}$ can be accomplished in any number of ways. One might calculate the characteristic polynomial and verify that it is independent of $k$ at flat points, or otherwise one might numerically calculate band structure at many different points and verify that the bands are indeed flat.

\subsubsection{Lotus lattices}

For each $\{n,3\}$ Schl\"afli symbol with $n \geq 6$, one can construct a graph by replacing regular polygons with ``tiles" constructed from $2n$ individual glued trees (or in the drawings above, shrubs).
There are, of course, many ways to do this.
Construct a regular polygon with $n \geq 6$ sides, and place a vertex at every corner. Place an additional vertex at the midpoint of each side. Then, place a vertex at the center of the polygon. Connect\footnote{By ``connect", we mean to identify one of the degree $p$ vertices in $S_p$ with the central vertex and the other vertex of degree $p$ with a vertex at the midpoint of a side.} a $p$-shrub from the central vertex to one lying at the midpoint of a side. Continue in this way, identifying the exterior vertices of adjacent $p$-shrubs. Finally, add vertices and edges in such a way that midpoint vertices that share a face are also connected by a $p$-shrub.
In this way, one can construct a regular polygon made of rhombi, which are all $p$-shrubs. We will affectionately refer to graphs constructed in this way as \textit{lotus tiles} of $n$ leaves since they evoke the form of the lotus flower.
We have drawn the first few lotus tiles in Figure \ref{fig:lotus} (d).

A lotus tile of $n$ leaves is made up of $2n$ $p$-shrubs. We can choose a gauge where every face in a given $p$-shrub winds the same phase, equal to either $\Phi$ or $-\Phi$, in such a way that the total phase wound around the boundary of the lotus tile vanishes.
The importance of this fact is difficult to overstate; given some shape that tiles either $\mathbb R^2$ or the Poincar\'e disk, it is possible to choose a gauge compatible with some tiling if and only if the total phase around the boundary of the tile vanishes. It, intuitively, is no loss of generality to require that the phase across any edge on the boundary of the tile is also zero.
This observation is of particular importance to lotus tiles with an odd number of sides.

We refer to lattices constructed by tessellating a space with lotus tiles in this way as \textit{lotus lattices}.
Some lotus lattices constructed from $3$-shrubs are drawn in Figure \ref{fig:lotus}.
In what follows, we write $H(n,p,\Phi)$ for a Hermitian complex-weighted adjacency matrix on the $\{n,3\}$ lotus lattice generated from $p$-shrubs whose restriction to the vertices of each $p$-shrub is gauge equivalent to $A^{\{p\}}(\Phi)$ or to $A^{\{p\}}(-\Phi)$, the sign being the one indicated by the coloring of that shrub in Figures \ref{fig:lotustrio} and \ref{fig:lotus}.
This restriction is a complex-weighted adjacency matrix of the shrub itself, since roots are never adjacent to roots and interior vertices are never adjacent to interior vertices.
Since every face of a lotus lattice is a face of some shrub, the flux through each face is fixed, and $H(n,p,\Phi)$ is determined up to a gauge transformation.
We will not specify a gauge, since it is unimportant in what follows.

Rather than including a precise definition, we content ourselves with studying a few examples.
It is straightforward to see how the above lattices could be generalized to $p \neq 3$, or greater depths, or a greater number of sides for the primitive polygon.

Note that Lemma \ref{lem:travel} does not directly apply to lotus lattices since they cannot be constructed by simply replacing the edges on some graph with glued trees.
Nonetheless, showing that Theorem \ref{thm:main} \emph{does} apply is possible.
In particular, lotus lattices are bipartite between the roots and the interior vertices of their shrubs, and for such lattices the hypothesis of Theorem \ref{thm:main} reduces to a computation involving only paths of length two.
\begin{lem}\label{lem:hubrim}
    Let $G = (E,V)$ be a bipartite connected graph of finite maximum degree $D$.
    Let $\mathcal H \cup \mathcal R = V$ and $\mathcal H \cap \mathcal R = \emptyset$ be a bipartition of $G$.
    If $H$ is a Hermitian complex-weighted adjacency matrix on $G$ with
    \begin{equation}\label{eqn:hubrim}
        \langle h' | H^2 | h \rangle = 0
    \end{equation}
    whenever $h$ and $h'$ are distinct vertices in $\mathcal H$,
    then $\langle a | H^m | b \rangle = 0$ for all $m \geq 1$ whenever $\mathrm{dist}(a,b) \geq 3$, so Theorem \ref{thm:main} applies with $N = 3$.
    Moreover, $\sigma(H) \subset \{0, \pm1, \pm\sqrt 2, \pm\sqrt 3, \dots, \pm\sqrt D\}$.
\end{lem}
\begin{proof}
    With respect to $\ell^2(V) = \ell^2(\mathcal H) \oplus \ell^2(\mathcal R)$,
    we have
    \begin{equation}
    H = \begin{pmatrix}
            0  & W \\
			W^\dagger & 0
    \end{pmatrix},
    \end{equation}
    so $H^2 = WW^\dagger \oplus W^\dagger W$.
    Since the nonzero entries of $H$ have unit modulus, $\langle h | H^2 | h\rangle = \deg(h)$, and (\ref{eqn:hubrim}) says exactly that $WW^\dagger = \sum_{h \in \mathcal H} \deg(h) |h\rangle\langle h|$.
    Hence $H^{2m}|h\rangle = \deg(h)^m | h\rangle$ and $H^{2m+1}|h\rangle = \deg(h)^m H|h\rangle$, so, with $B_r(b)$ as in the proof of Theorem \ref{thm:main}, $H^m|h\rangle$ is supported in $B_1(h)$ for every $h$ in $\mathcal H$. For $r$ in $\mathcal R$, $H|r\rangle$ is supported on the vertices in $\mathcal H$ adjacent to $r$, and thus  $H^m|r\rangle$ is supported entirely on $B_2(r)$.
    This proves that Theorem \ref{thm:main} applies with $N = 3$. For the spectral claim, $H$ is bounded by the Schur test, as in the proof of Theorem \ref{thm:main}, and $\sigma(W^\dagger W) \cup \{0\} = \sigma(WW^\dagger) \cup \{0\}$ for bounded $W$. Hence $\sigma(H^2) \subseteq \sigma(WW^\dagger) \cup \{0\} = \{\deg(h) \,:\, h \in \mathcal H\} \cup \{0\} \subseteq \{0, 1, \dots, D\}$, and since $\lambda \in \sigma(H)$ implies $\lambda^2 \in \sigma(H^2)$, the claim follows.

\end{proof}
Theorem \ref{thm:main} is still needed to produce a compactly supported basis of $\ker H$.
\begin{figure}
\centering
\begin{subfigure}{0.32\linewidth}
    \includegraphics[width=\linewidth]{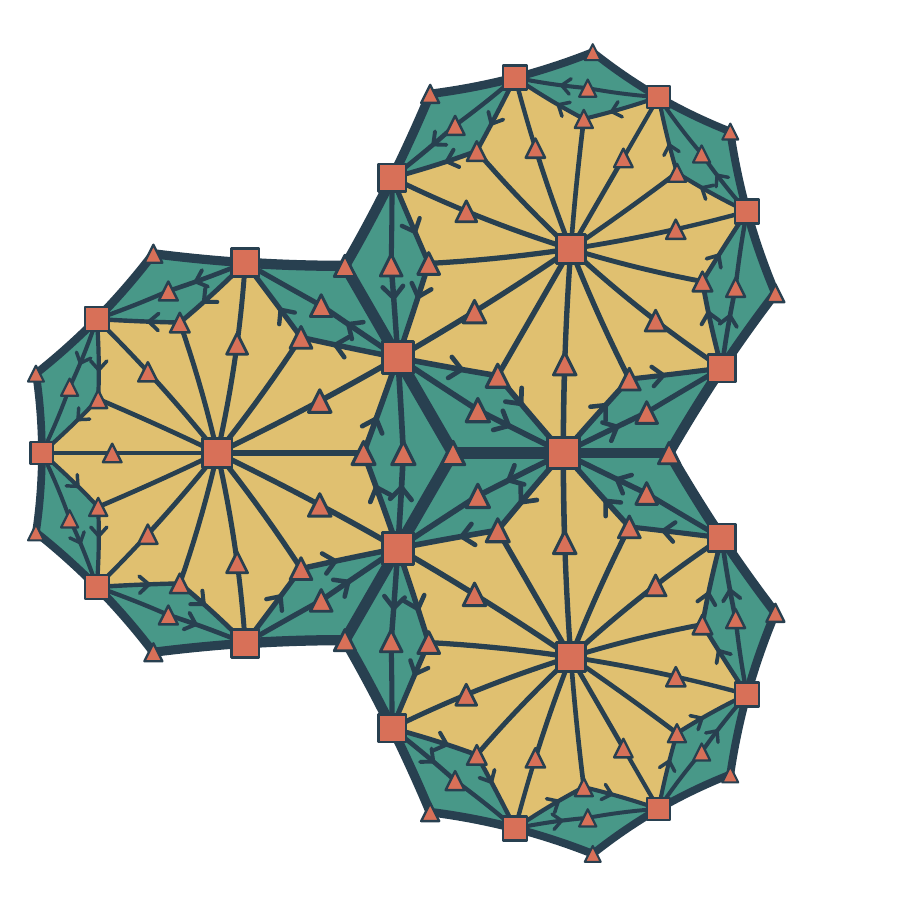}
    \caption{}
\end{subfigure}
\begin{subfigure}{0.32\linewidth}
    \includegraphics[width=\linewidth]{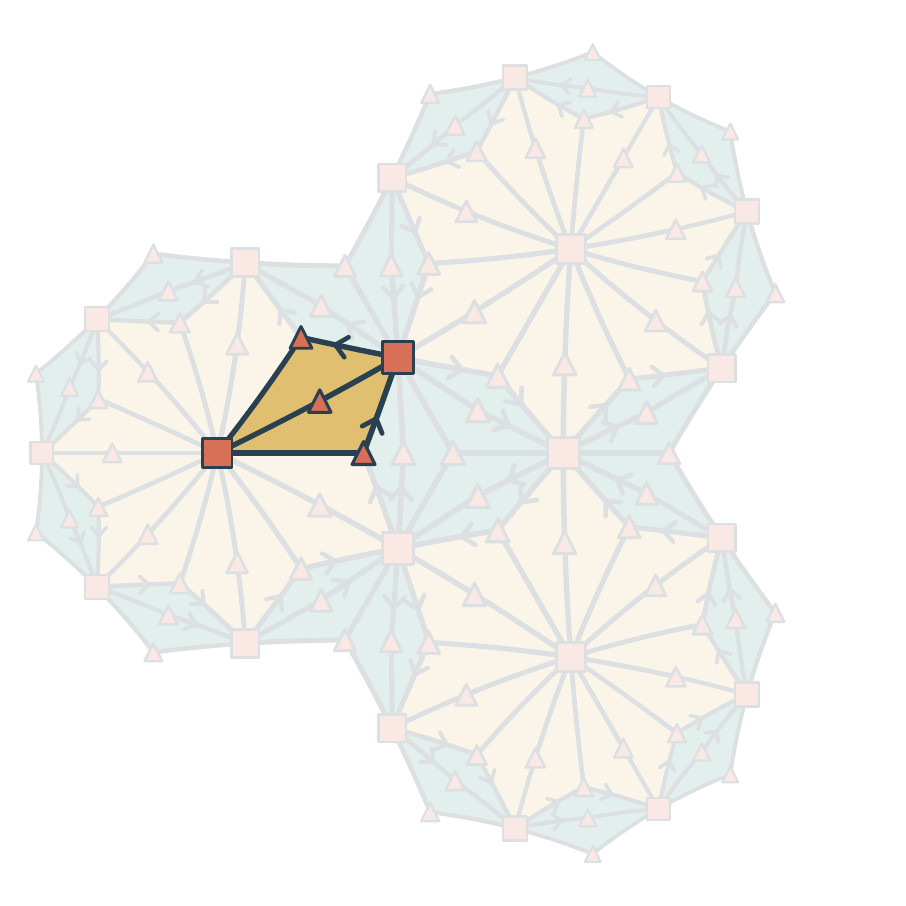}
    \caption{}
\end{subfigure}
\begin{subfigure}{0.32\linewidth}
    \includegraphics[width=\linewidth]{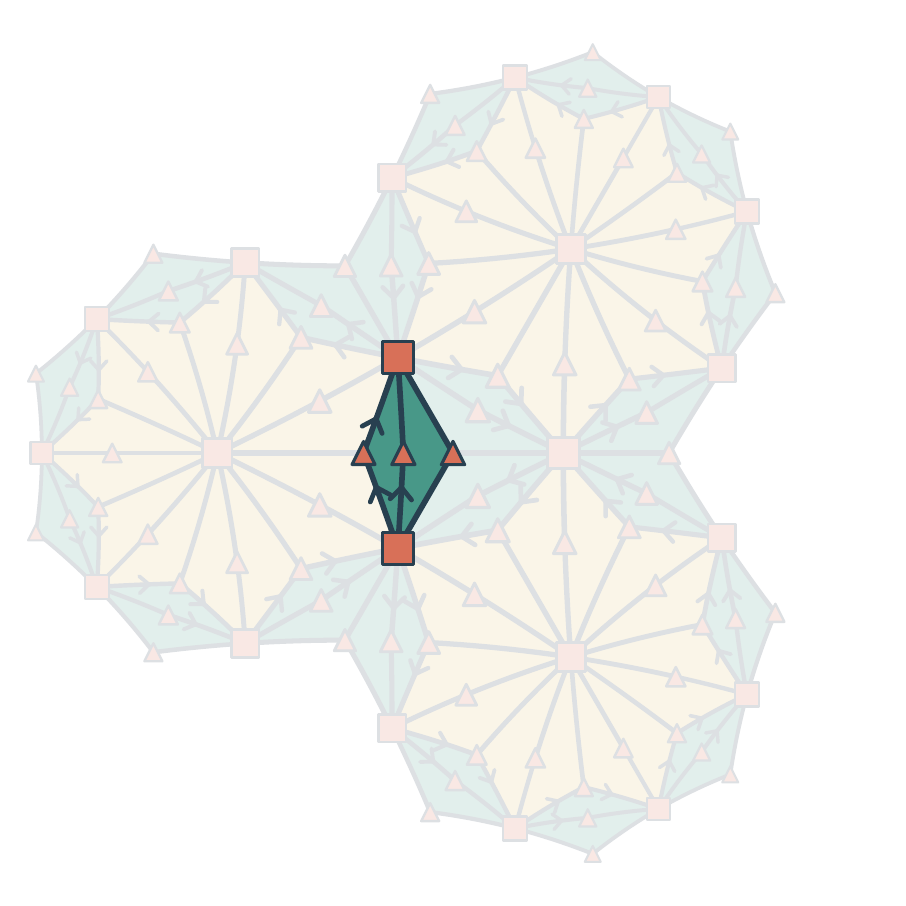}
    \caption{}
\end{subfigure}
\caption{The tiles of the $\{7,3\}$ lotus lattice generated from $3$-shrubs, meeting at a common corner. Vertices decorated with squares are  members of $\mathcal H$ and vertices decorated with a triangle are members of $\mathcal R$ as defined in (\ref{eqn:partitionz}). An arrow between vertices $i$ and $j$ denotes that $\langle j | H(n,p,\Phi) | i \rangle = \mathrm{e}^{\mathrm i \Phi }$, and plaquettes are colored by the sign of the phase wound around their boundary. }
    \label{fig:lotustrio}
\end{figure}
\begin{figure}
\centering
\begin{subfigure}{0.24\linewidth}
    \includegraphics[width=\linewidth,angle=90]{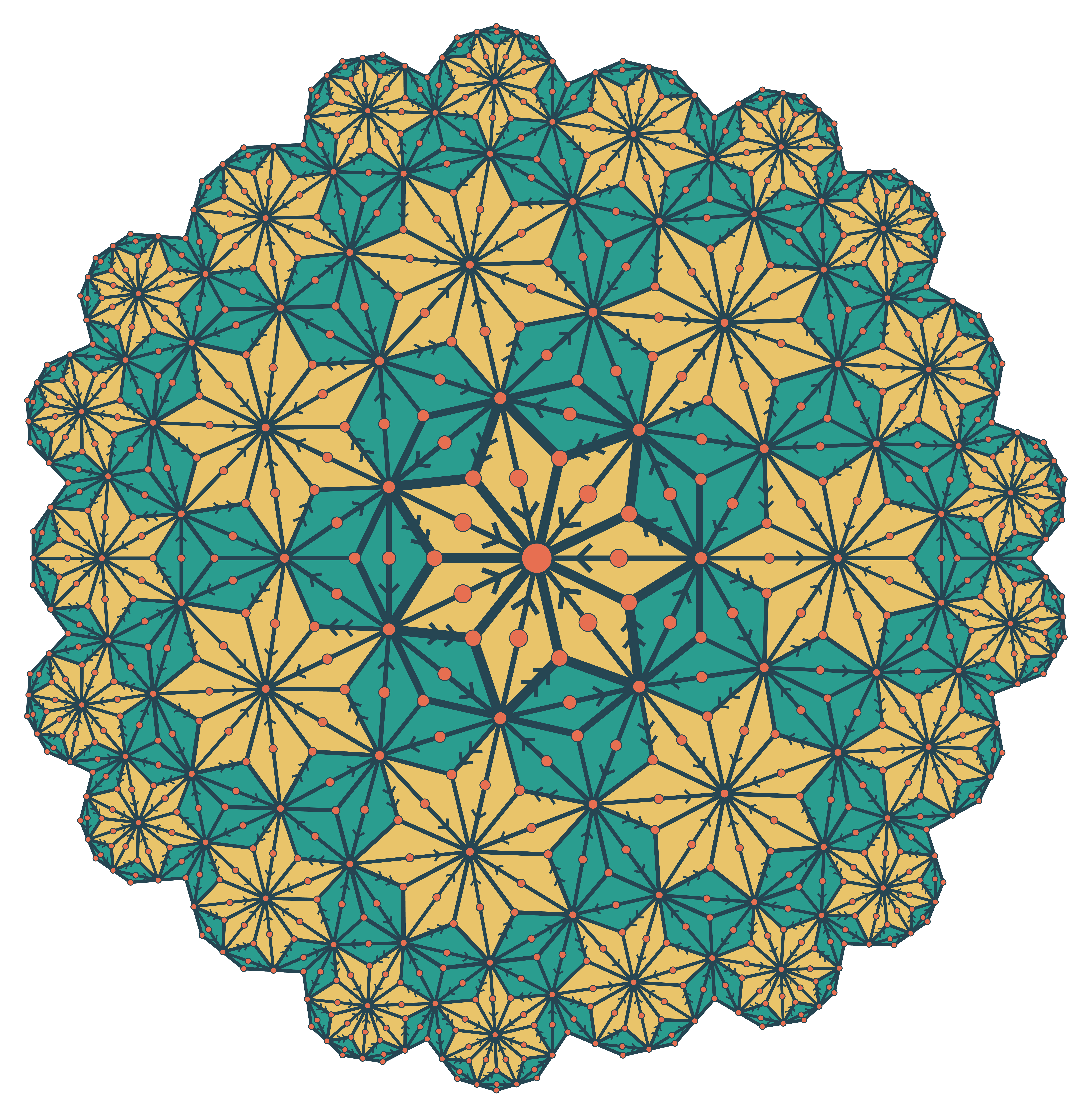}
    \caption{The $\{7,3\}$ lotus lattice generated from $3$-shrubs.}
\end{subfigure}
\begin{subfigure}{0.24\linewidth}
    \includegraphics[width=\linewidth]{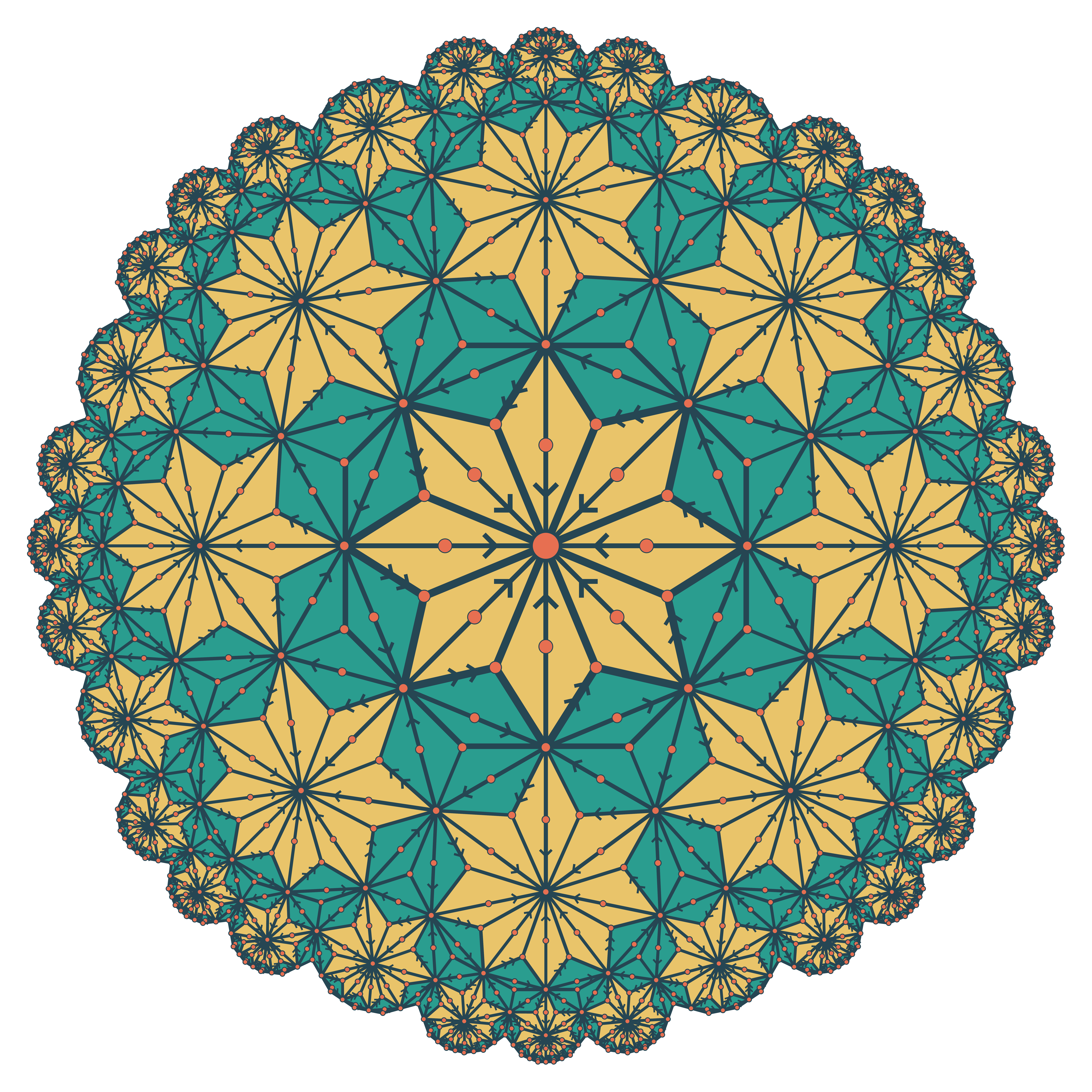}
    \caption{The $\{8,3\}$ lotus lattice generated from $3$-shrubs.}
\end{subfigure}
\begin{subfigure}{0.24\linewidth}
    \includegraphics[width=\linewidth,angle=90]{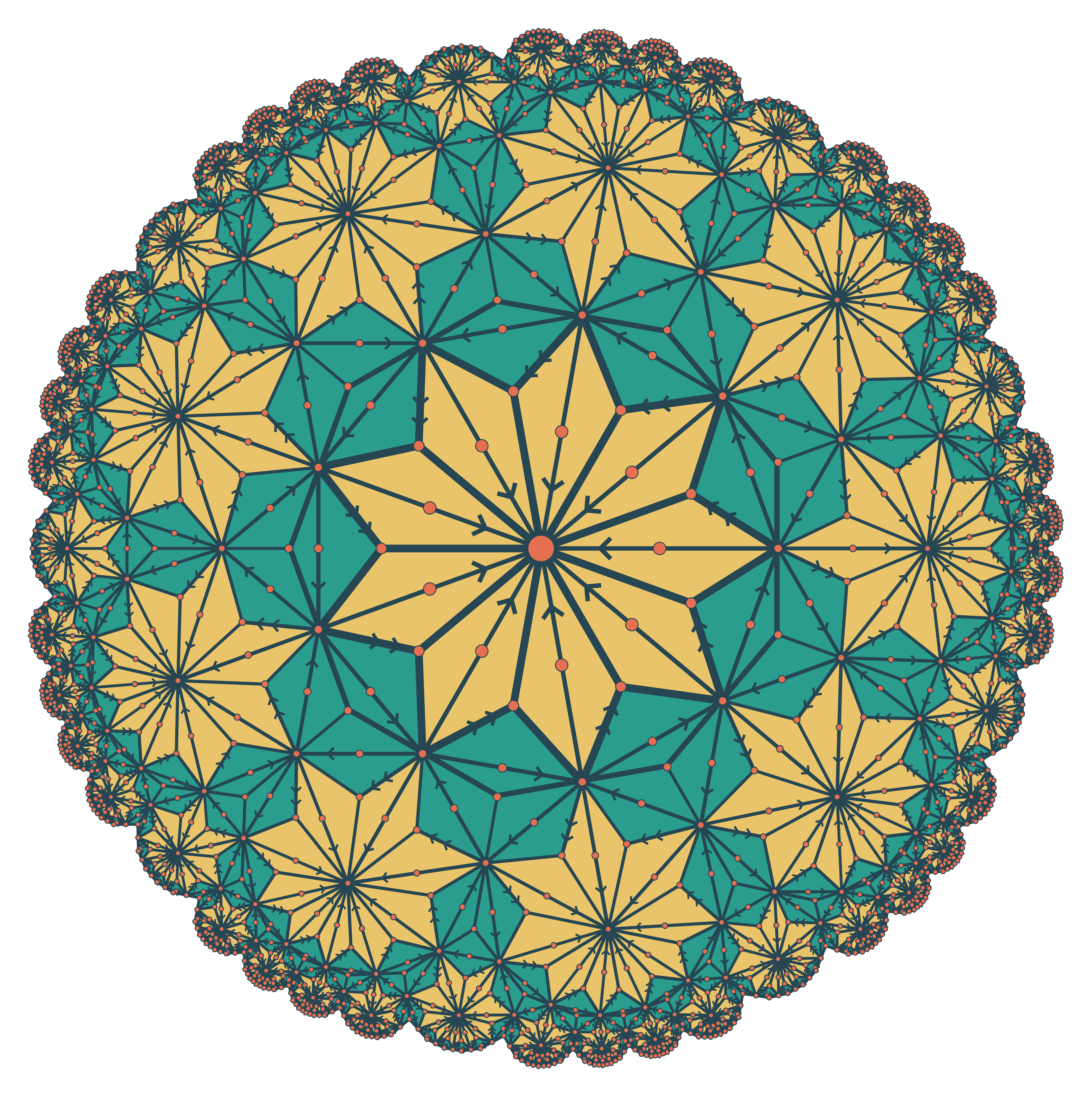}
    \caption{The $\{9,3\}$ lotus lattice generated from $3$-shrubs.}
\end{subfigure}
\begin{subfigure}{0.22\linewidth}
    \includegraphics[width=\linewidth]{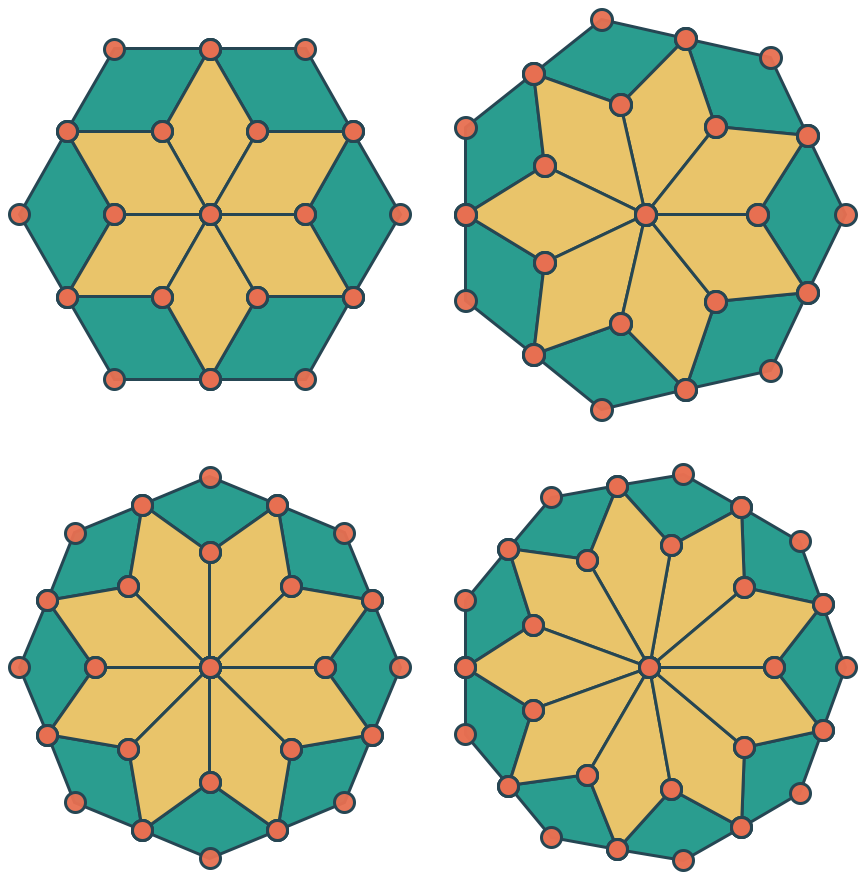}
    \caption{Various lotus tiles generated from $2$-shrubs.}
\end{subfigure}
    \caption{Pictured above are finite chunks of various lotus lattices constructed from $3$-shrubs drawn on the Poincar\'e disk. A dark arrow between a pair of vertices $i$ and $j$ indicates that $\langle i | H(n,p,\Phi) |j\rangle =  e^{\mathrm i \Phi}$, while a pair of dark arrows indicates that $\langle i | H(n,p,\Phi) | j \rangle = e^{2 \mathrm i\Phi}$. Plaquettes are colored according to the sign of phase wound counterclockwise around their boundary.
    }
    \label{fig:lotus}
\end{figure}

To show that Theorem \ref{thm:main} applies, it is only necessary to use Lemma \ref{lem:hubrim} by producing a bipartition of vertices $V = \mathcal H \cup \mathcal R$  such that
\begin{equation}\label{eqn:condition}
   \langle h'|H(n,p,\Phi)^2 |h \rangle = 0,
\end{equation}
so long as $\Phi \in F^{\{p\}}$, for every pair of distinct $h$ and $h'$ in $\mathcal H$.

Let $G$ be a $\{n, 3\}$ lotus lattice generated from $p$-shrubs, and consider the vertex partition
\begin{equation}\label{eqn:partitionz}
\begin{aligned}
	\mathcal H &= \{h\, :\, \deg(h) > 3\} \\
	\mathcal R &= \{h\, :\, \deg(h) \leq 3\}.
	\end{aligned}
\end{equation}
Vertices in $\mathcal H$ are either the central vertex of a tile with degree $n(p-1)$ or a vertex on the midpoint of a side of a lotus tile of degree $6(p-1)$. Every edge of $G$ is an edge of some $p$-shrub joining a root to an interior vertex; the shrub roots are precisely the tile centers and side midpoints. Interior vertices have degree two or three, so every edge joins $\mathcal H$ to $\mathcal R$, as Lemma \ref{lem:hubrim} requires. Distance two paths between a pair of vertices in $\mathcal H$ are thus one of two kinds; either between the central vertex of a tile and one of its sides, or between side-midpoints sharing a corner in the lotus tile.
The corners of lotus tiles are also degree three, and are the common leaf of three distinct shrubs, each strictly within their own lotus tile, so pairs of vertices in $\mathcal H$ at distance two are vertices within a single tile.
In either case, the common neighbors of $h$ and $h'$ are the interior vertices of  their common shrub, within a single lotus tile.
For concreteness, see Figure \ref{fig:lotustrio} where $n = 7$, and $p = 3$, although our discussion is largely agnostic to $n$ and $p$. The two differing pairs of vertices along with their common neighbors are depicted in Figure \ref{fig:lotustrio}(b)-(c).
Since only paths of length two contribute to (\ref{eqn:condition}), the question of whether (\ref{eqn:condition}) holds can be addressed by considering the complex-weighted adjacency matrix consisting of $h$, $h'$ and their common neighbors. Such subgraphs are always $p$-shrubs, by construction, so (\ref{eqn:condition}) holds when every face of the shrub winds the same  $\Phi$  and $\Phi$ is among the flat points $F^{\{p\}}$ of the $p$-shrub: the restriction of $H(n,p,\Phi)$ to such a shrub is gauge equivalent to $A^{\{p\}}(\pm\Phi)$ by the definition of $H(n,p,\Phi)$, so (\ref{eqn:condition}) is Corollary \ref{cor:nocross} with $m = 2$.
Helpfully, Lemma \ref{lem:hubrim} also provides that
\begin{equation}
	\sigma(H(n,p,\Phi)) \subset \{0, \pm\sqrt{6(p-1)}, \pm\sqrt{n(p-1)}\},
\end{equation}
provided again that $\Phi \in F^{\{p\}}$.

Now we are in a position to appreciate that Theorem \ref{thm:main} did not rely upon notions of Bloch's theorem.
It is possible to exploit various Fuchsian boosts in order to periodically identify the boundary of a finite chunk of lotus lattices, in such a way that states living in a one-dimensional irreducible representation of some such boost are fully captured, but doing so requires substantial algebraic technology \cite{hypbloch,hyperblochmath}.
A proof of the same statement from the perspective of hyperbolic Bloch theory may amount to an interesting algebraic result, but we defer this to future work.

One may also wonder about other $\{p,q\}$ tilings of lotus tiles. For $q \neq 3$, there are generally states that are not compact localized. To see that this is the case, construct the $\{4,4\}$ lotus lattice out of any desired $p$-shrub and observe that there are horizontal and vertical lines along which particles may freely traverse. An analogous phenomenon occurs for $\{p,q\}$ lotus lattices with $q > 3$.
Finally, one may wonder about the $\{6,3\}$ lotus lattice, since it apparently satisfies our criterion.
It turns out that the $\{6,3\}$ lotus lattice constructed from $2$-shrubs is precisely the dice lattice. Hence, we have found infinitely many generalizations of the dice lattice that we expect to be phenomenologically similar to the dice lattice itself.
We have drawn one such generalization of the dice lattice in Figure \ref{fig:thrice}.
We will discuss the dice lattice and generalizations thereof in greater detail in the next section.

The attentive reader will also wonder whether a tree grown by an integer sequence $X = \{x_1, \dots, x_d\}$ may also be used to construct lotus lattices. In short, the answer is yes, although only the point $\Phi = \frac{2 \pi}{x_1}$ admits flat bands and only compact localized states. The reason for this is effectively that the eigenstates at other flat points of a glued tree have support on a majority of a tree, rather than exactly half of it.
This freedom to cross most but not all of a glued tree is sufficient to spoil the perfect confinement necessary for Theorem \ref{thm:main} to hold.

\begin{figure}
    \centering
        \includegraphics[width=0.7\linewidth]{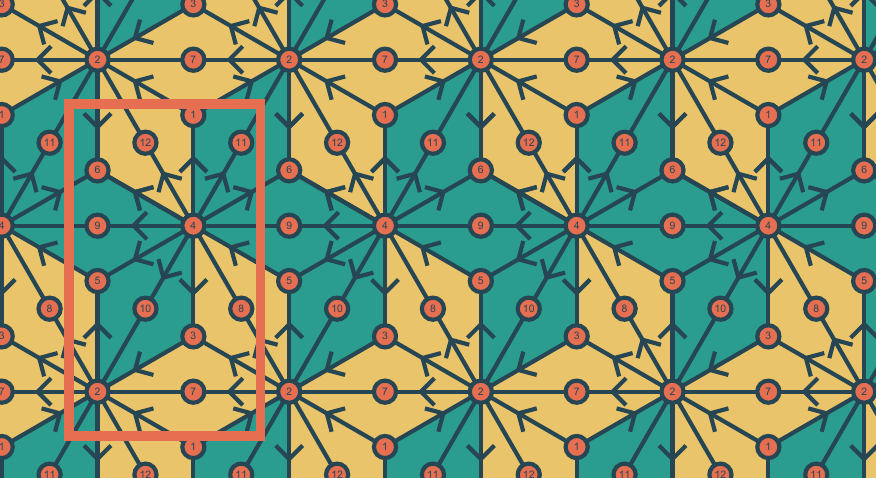}
    \caption{A diagram of the $X = \{3\}$ generalization of the dice lattice. Green and yellow plaquettes are colored according to the sign of the phase wound clockwise around the plaquette. An orange box is drawn around a single unit cell. An arrow on the edge from site $i$ to site $j$ indicates that $\langle i | A| j \rangle = e^{\mathrm i \Phi}$. Note that this lattice is not made of glued trees in the canonical gauge. }
    \label{fig:thrice}
\end{figure}

\subsubsection{Another kind of lotus lattice}
\begin{figure}
    \centering
    \begin{subfigure}{0.24 \linewidth}
        \centering
        \includegraphics[height=\linewidth,angle=0]{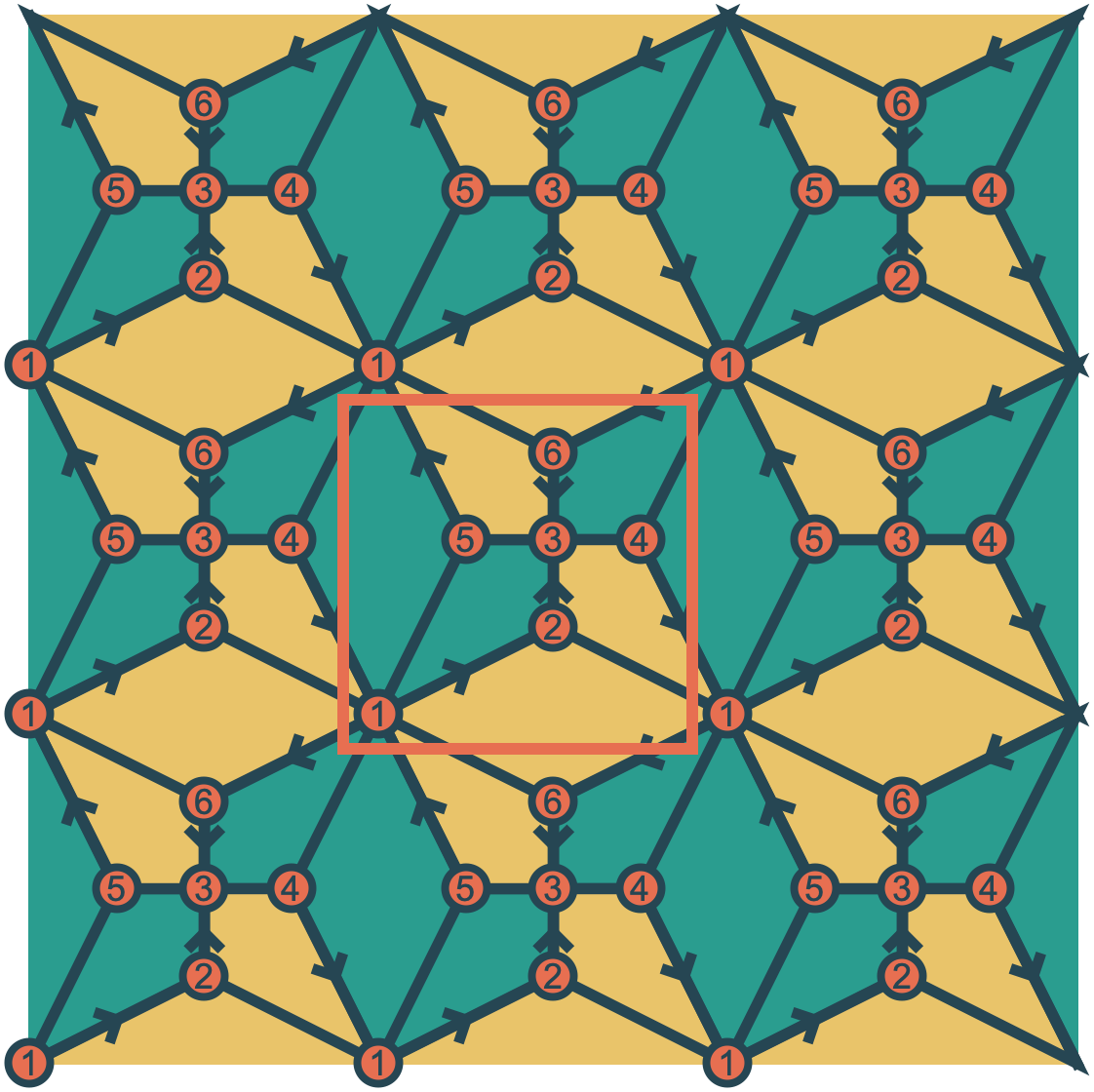}
        \caption{\{4,4\} lotus lattice of the second kind.}
    \end{subfigure}
    \begin{subfigure}{0.24\linewidth}
        \centering
        \includegraphics[width=\linewidth]{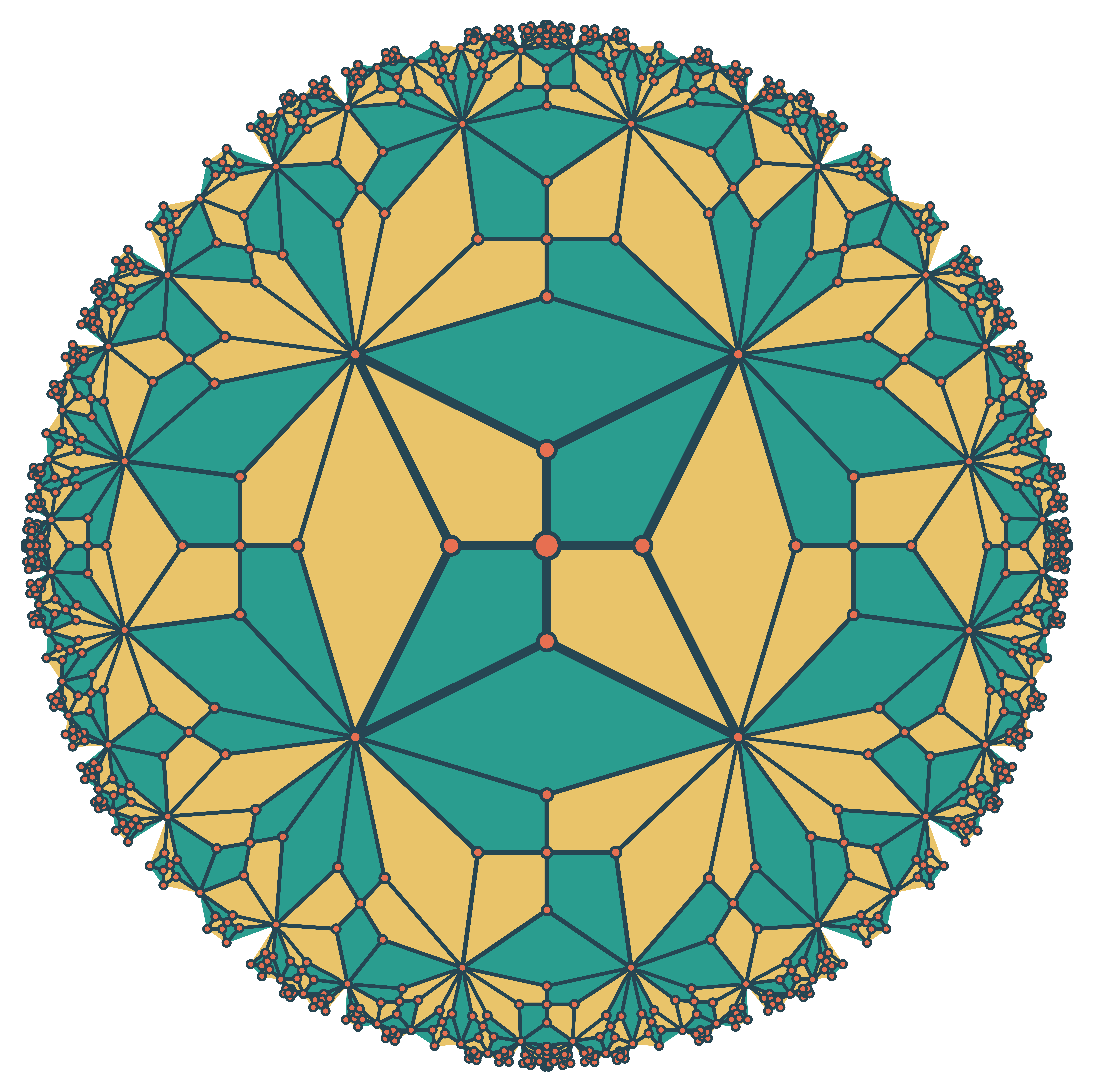}
        \caption{$\{4,6\}$ lotus lattice of the second kind.}
    \end{subfigure}
    \begin{subfigure}{0.24\linewidth}
       \centering
       \includegraphics[width=\linewidth]{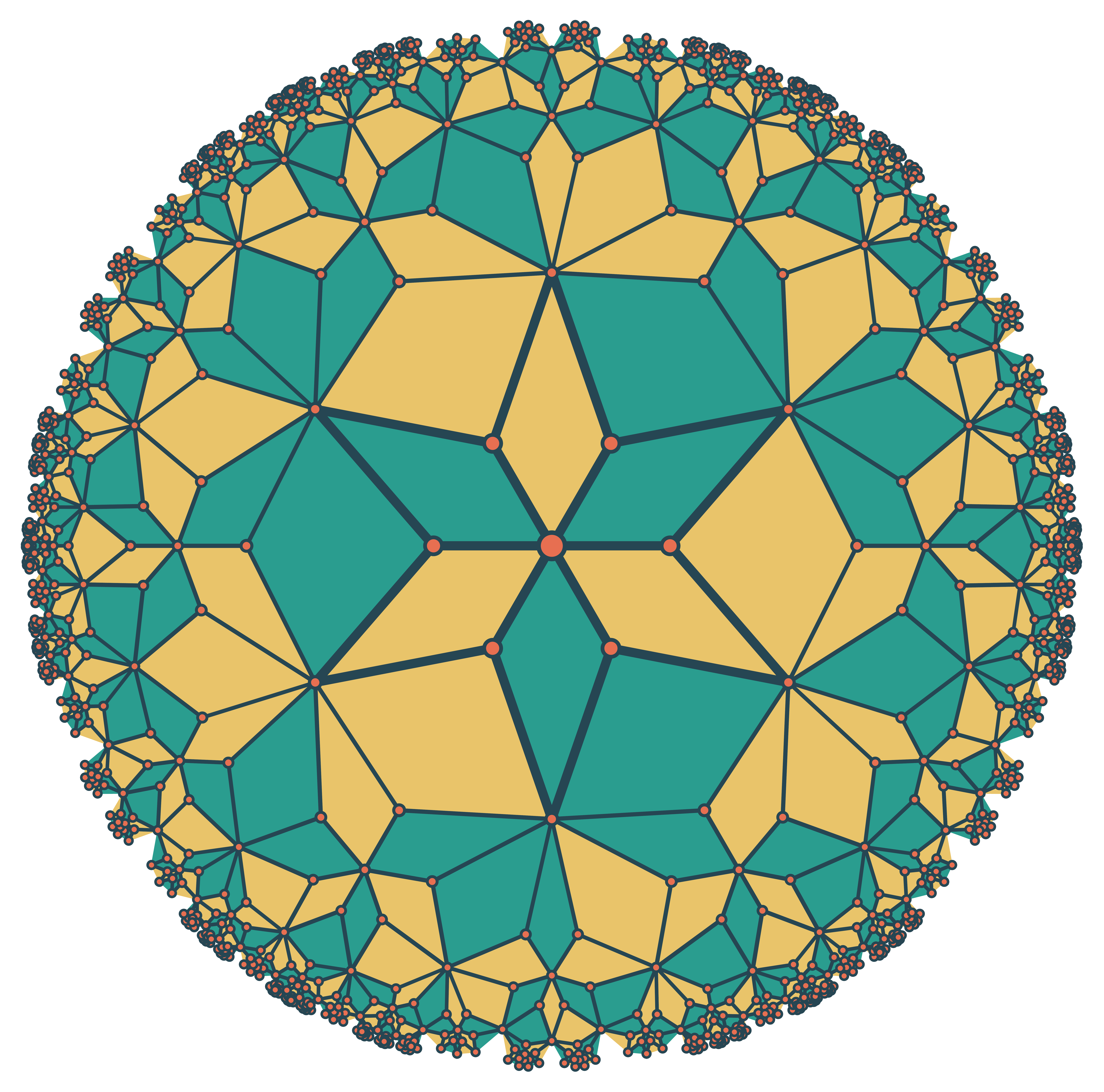}
       \caption{$\{6,4\}$ lotus lattice of the second kind. }
    \end{subfigure}
    \begin{subfigure}{0.24\linewidth}
       \includegraphics[width=\linewidth]{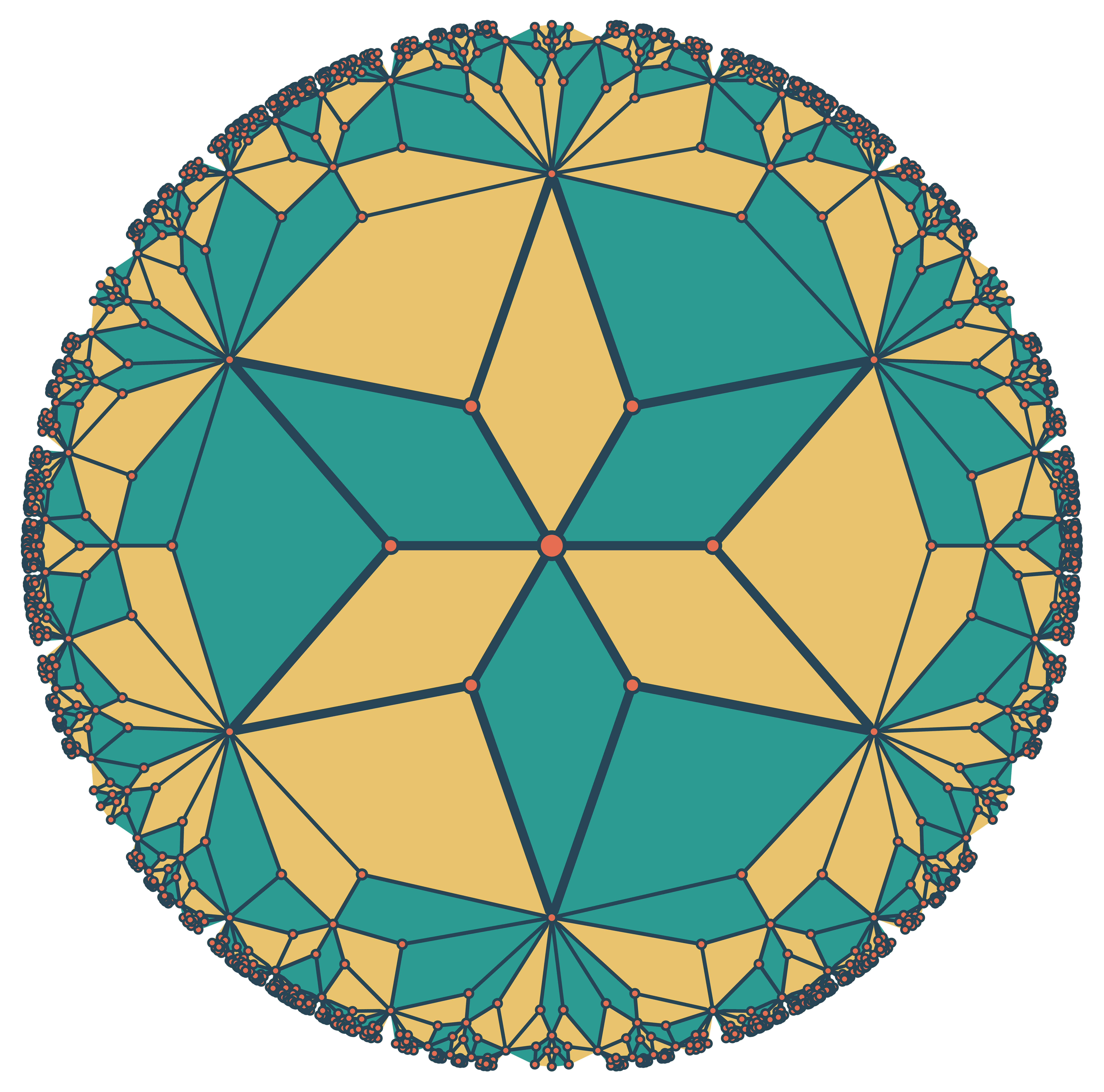}
       \caption{$\{6,6\}$ lotus lattice of the second kind. }
    \end{subfigure}
    \caption{Again, plaquettes above are colored according to the sign of the phase wound around the boundary of the plaquette. In (a), we have drawn a unit cell boundary in orange, and arrows according to a particular choice of gauge consistent with the coloring presented. }\label{fig:secondkind}
\end{figure}

The dice lattice is a rich physical model. It is known to have a number of interesting properties.
Among these are a flat-band structure when edges carry the right complex weights \cite{debnath2024magnonsdicelatticetopological}, Majorana corner modes \cite{majorana}, and nontrivial band topology with the appropriate next-nearest-neighbor hopping terms \cite{dicehaldane}.
Here, we present infinitely many generalizations of the dice lattice that we expect to be similarly rich.

Consider a regular $2n$-gon. Place vertices at every corner and at the center of the $2n$-gon. Once again, connect a $p$-shrub from the central vertex to each of the corner vertices, and identify a vertex of degree two in one $p$-shrub with a vertex of degree two in the next $p$-shrub clockwise.
For simplicity, we will restrict our attention to $p = 2$.
To see the graphs of which we speak, consider Figure \ref{fig:lotus} but remove every vertex of degree $2$ from the drawings therein.

The corresponding ``stars" can be tiled as regular polyhedra in any $\{2n, p\}$ tiling of space. Of course, not all of these tilings lead to a lattice that can be shown to have only compact localized states without further caveats. We restrict our attention to only $\{2n,2p\}$ tilings of space.
We name the family of lattices constructed in this way \textit{lotus lattices of the second kind}.
We have pictured a few of them in Figure \ref{fig:secondkind}.
By reusing (\ref{eqn:partitionz}) and again attempting to apply Lemma \ref{lem:hubrim}, it is possible to check that Theorem \ref{thm:main} applies to all of the lattices pictured in a suitable gauge.

Of particular interest to us, however, is the Euclidean lattice in Figure \ref{fig:secondkind} (a), which we have been unable to find in existing literature.
It has a simple complex-weighted adjacency matrix, which gives rise to the Bloch Hamiltonian
\begin{equation}\label{eqn:new}
H =
    \begin{pmatrix}
    0 & e^{-\mathrm i k_x} + \omega & 0 & e^{-\mathrm i (k_x + k_y) } + \omega^*  e^{-\mathrm i k_x} & 1 + e^{-\mathrm i k_y } \omega^* & e^{-\mathrm i  k_y} + e^{-\mathrm i (k_x + k_y) } \omega \\
    e^{\mathrm i k_x} + \omega^* & 0 & \omega & 0 & 0 & 0 \\
    0 & \omega^* & 0 & 1& 1& \omega^* \\
    e^{\mathrm i (k_x + k_y)} + e^{\mathrm i k_x}\omega & 0 & 1& 0 & 0& 0 \\
    1 + e^{\mathrm i k_y} \omega & 0 & 1& 0& 0& 0 \\
    e^{\mathrm i k_y} + e^{\mathrm i (k_x + k_y)} \omega^* & 0 & \omega & 0 & 0 & 0
    \end{pmatrix}
    \end{equation}
where $\omega = e^{\mathrm i \Phi/2}$ for some real $\Phi$.
$H$ in (\ref{eqn:new}) has a small enough unit cell to be studied in near term quantum simulation experiments, and may have interesting physical properties.
We explicitly apply Lemma \ref{lem:hubrim} using (\ref{eqn:partitionz}).
We have $\mathcal H = \{1,3\}$ and $\mathcal R = \{2,4,5,6\}$.
By translation invariance, the hypothesis (\ref{eqn:hubrim}) that $\langle h' | H^2 | h\rangle$ vanish for every pair of distinct hubs is equivalent to the statement that the $\mathcal H$  block of $H^2$ equals $\mathrm{diag}(8,4)$ for every $k$. The off-diagonal entry of this block is
\begin{equation}
    \langle 3 | H^2 | 1 \rangle = \left(\omega + \omega^*\right)\left(e^{\mathrm i k_x} + e^{\mathrm i k_y}\right) + \left(1 + \omega^{*2}\right)\left(1 + e^{\mathrm i (k_x + k_y)}\right),
\end{equation}
which vanishes for all $k_x$ and $k_y$  if and only if $\omega^2 = -1$, or equivalently $\Phi = \pi$.
In this case, the diagonal entries are $8$ and $4$, while for other $\Phi$ the entry $\langle 1 | H^2 | 1\rangle$ depends on $k_x$ or $k_y$ .
Lemma \ref{lem:hubrim} therefore immediately gives the spectrum $\{0, \pm 2, \pm 2\sqrt 2\}$ at $\Phi = \pi$.

Why have we restricted to only $\{2n, 2q\}$ lotus lattices of the second kind?
The reason for this is that the regular $2n$-gons we use for tiling contain $2n$ $p$-shrubs and $n$ halves of a $p$-shrub. In order to admit a choice of gauge that is consistent with desired translation symmetries, the total flux piercing the $2n$-gon must vanish, and the half-shrubs on opposite sides of the $2n$-gon must be pierced by the same flux. From this it follows that one cannot find a gauge in general that is compatible with arbitrary tilings of any $(2n+1)$-gon. To see that odd-degree tilings give a similar problem, observe that additional constraints on the coloring of plaquettes are imposed by the smallest possible composition of translation symmetries that compose to the identity.

We have restricted the second argument to be even for simplicity.
$\{2n,p\}$ lotus lattices of the second kind can be shown to have only compact localized states at shrub flat points even if there does not exist a choice of phase consistent with the periodicity of the adjacency matrix of the graph in question.
We defer a thorough treatment of such structures to future work on the grounds that it would be inconvenient to introduce the technology required to study such lattices, which are intrinsically quasicrystalline.
We remark in passing that $\{4n, p\}$ lotus lattices of the second kind \emph{do} generally admit colorings that are consistent with any number of translation symmetries.
We do not claim to offer an exhaustive list of lotus lattices that have only compact localized states at flat points of constituent shrubs.

\section{Conclusions}
In this paper, we have introduced a broadly applicable framework for the construction and analysis of lattices that, at the level of noninteracting single-particle dynamics, have only compact localized eigenstates (and flat bands) induced by Aharonov--Bohm caging.
For every countably infinite graph, we have introduced one such model for every string of positive integers of finite length corresponding to the replacement of edges with glued trees.
We have also introduced a procedure whereby tessellations of glued trees in space give rise to models with the same exotic properties, almost all of which are non-Euclidean lattices.
In this paper, we do not claim to have enumerated all lattices with Aharonov--Bohm induced locality but it would be interesting to consider whether some such classification is possible.
Another question of interest is whether there is some connection between localization in this context and localization induced by, say, disorder.

This paper does not focus on the physical properties of the models we introduce beyond perfect localization,
but it would be interesting to study (for example) ground state phase diagrams of the lattices we have studied in the presence of on-site interactions or with additional hopping terms, where one might reasonably expect topological bands or various exotic edge states to emerge.
Physically motivated study along these lines is of particular interest on the grounds that some of the lattices we introduce may be simulable in near-term cold atom or superconducting circuit experiments.

\section{Acknowledgments}

This work was supported by the National Science Foundation Quantum Leap Challenge Institute for Robust Quantum Simulation (2120757) and the Co-design Center for Quantum Advantage (C$^2$QA) under Contract No. DESC0012704.
We acknowledge useful and interesting conversations with Sarang Gopalakrishnan, Andrew Lucas, Basil Smitham, and Jeronimo Martinez.

\begin{appendix}
    \section{Glued tree spectra from characteristic polynomial recurrence relations}\label{app:chrpolyrec}
In this appendix we discuss one of two possible ways to derive the eigenvalues of a glued tree. We emphasize that we do not consider lattices of glued trees until Appendix \ref{app:chrpolyafb}.
\begin{lem}\label{lemma:huge}
    Let $X = \{x_1, x_2, \dots, x_d\}$ be a set of positive integers with $x_d > 1$.
    Define $X_i = \{x_1, x_2, \dots, x_i\}$ for $i \leq n$.
    As a shorthand, write $Y_i = A^{X_i}(\Phi)$ to denote the CCAM of the glued tree grown by $X_i$.
    We have a number of identities for each $1 \leq i \leq d$:
    \begin{equation}\label{eqn:claims}
        \begin{aligned}
            C(Y_i;\lambda) \langle f_{Y_i} |  R(Y_i;\lambda)| f_{Y_i}\rangle &=  -C(Y_{i-1};\lambda)^{x_i}(\lambda + x_i \langle l_{Y_{i-1}} |  R(Y_{i-1};\lambda)| l _{Y_{i-1}}\rangle) \\
            C(Y_i;\lambda) \langle l_{Y_i} | R (Y_i;\lambda)| l_{Y_i}\rangle &=  -C(Y_{i-1};\lambda)^{x_i}(\lambda + x_i \langle f_{Y_{i-1}} |  R(Y_{i-1};\lambda)| f _{Y_{i-1}}\rangle) \\
             C(Y_i;\lambda)\langle l_{Y_i} |   R(Y_i;\lambda)| f_{Y_i}\rangle &= \sigma C(Y_{i-1};\lambda)^{x_i} \langle f_{Y_{i-1}} |  R(Y_{i-1};\lambda)| l_{Y_{i-1}}\rangle  (\langle \omega_i|)^*| \omega_i\rangle
        \end{aligned}
    \end{equation}
    with $\sigma \in \{-1,1\}$.
\end{lem}
\begin{proof}
The identities above are most naturally shown by manipulating the adjugate of $\Lambda_i := Y_i - \lambda$.
Note the identities
    \begin{equation}
        \begin{aligned}
           \text{det}  \begin{pmatrix}
               A & |x \rangle  \\
               \langle y |  & \alpha
           \end{pmatrix} = \alpha \det(A) - \langle y |  A^{\mathrm A} | x\rangle, \\
           \begin{pmatrix}
             A & 0 \\
             0 & B
           \end{pmatrix}^{\mathrm A} =
           \begin{pmatrix}
               \text{det}(B) A^{\mathrm A} & 0 \\
               0 & \text{det}(A) B^{\mathrm A}
           \end{pmatrix}.
        \end{aligned}
    \end{equation}
    From this point, we proceed by direct calculation.
    Explicitly,
    \begin{equation}\label{eqn:diag}
    \begin{aligned}
        \langle f_{Y_i} | \Lambda_i^{\mathrm A}|f_{Y_i} \rangle &= \det \begin{pmatrix}
            \mathbb I_{x_i} \otimes \Lambda_{i-1} &  | \omega_i \rangle \otimes | l _{Y_{i-1}} \rangle  \\
            \langle \omega_i | \otimes \langle l _{Y_{i-1}} |& -\lambda
        \end{pmatrix} \\
        &= -\lambda \det \left[\mathbb I_{x_i} \otimes \Lambda_{i-1}\right] - \bra{\omega_i} \otimes \bra{l_{Y_{i-1}}} \left[\mathbb I_{x_i} \otimes \Lambda_{i-1}\right]^{\mathrm A} \ket{\omega_i}\otimes \ket{l_{Y_{i-1}}} \\
        &= -\lambda C(Y_{i-1};\lambda)^{x_i} - C(Y_{i-1};\lambda)^{x_i -1} \langle \omega_i| \omega_i \rangle \langle l_{Y_{i-1}} | \Lambda_{i-1}^{\mathrm A} | l_{Y_{i-1}}\rangle.
        \end{aligned}
    \end{equation}
    Now, $\langle \omega_i | \omega_i \rangle = x_i$, and
    \begin{equation}
       \Lambda_{i-1} ^{\mathrm A} = (Y_{i-1} - \lambda )^{\mathrm A} = C(Y_{i-1} ;\lambda) R(Y_{i-1};\lambda)
    \end{equation}
    so the first line of (\ref{eqn:claims}) is established.
    The second line of (\ref{eqn:claims}) follows analogously.

    As for the third line of (\ref{eqn:claims}), we have
    \begin{equation}\label{eqn:offdiag}
    \begin{aligned}
        \bra{l_{Y_i}}\Lambda_i^{\mathrm A}\ket{f_{Y_i}} &= \det \begin{pmatrix}
            \bra{\omega_i}^*\otimes \bra{f_{Y_{i-1}}} & 0  \\
            \mathbb I_{x_i }\otimes \Lambda_{i-1} & \ket{\omega_i}\otimes \ket{l_{Y_{i-1}}} \\
        \end{pmatrix} \\
        &=
         \sigma\det \begin{pmatrix}
            \mathbb I_{x_i}\otimes \Lambda_{i-1} & \ket{\omega_i}\otimes \ket{l_{Y_{i-1}}} \\
            \bra{\omega_i}^*\otimes \bra{f_{Y_{i-1}}} & 0  \\
        \end{pmatrix} \\
        &= \sigma  C(Y_{i-1};\lambda)^{x_i-1} \bra{f_{Y_{i-1}}}\Lambda_{i-1}^{\mathrm A} \ket{l_{Y_{i-1}}} (\langle \omega_i|^*) | \omega_i \rangle
        \end{aligned}
    \end{equation}
    with $\sigma = 1$ if $Y_i$ is an even-dimensional matrix and $\sigma = -1$ otherwise.
    \end{proof}

    Now, we set about establishing that (\ref{eqn:claims}) are significantly simplified by a particular symmetry.
\begin{defn}[Exchange matrix]\label{defn:exchangematrix}

    Let $N $ be a positive integer and $B$ be an $N$ by $N$ matrix.
    Define
    \begin{equation}
        J_B = \sum_{i = 1}^N |i \rangle \langle N - i |
    \end{equation}
    which we call the \textbf{exchange matrix}.
\end{defn}
We have chosen to subscript the exchange matrix with yet another matrix for the sake of emphasizing that $J_X$ is compatible with $X$ in the sense that they may be multiplied together. This choice is convenient because of various recurrence relations involving matrices of different sizes.
If $A$ and $B$ are square matrices, note that
\begin{equation}
    J_{A \otimes B} = J_A \otimes J_B.
\end{equation}
\begin{cor}\label{cor:symm}
Adopt the hypothesis of Lemma \ref{lemma:huge}.
The following results hold
\begin{enumerate}
    \item If $[Y_i, J_{Y_i}] = 0$, then $\langle l_{Y_i} | \Lambda_{i}^{\mathrm A} | l_{Y_i} \rangle = \langle f_{Y_i} | \Lambda_i^{\mathrm A} |f_{Y_i}\rangle$ and $\langle l_{Y_i} | \Lambda_i^{\mathrm A} | f_{Y_i} \rangle = \langle f_{Y_i} | \Lambda_i^{\mathrm A} | l_{Y_i} \rangle$.
    \item If $[Y_{i-1}, J_{Y_{i-1}}] = 0$ and $ J_{|\omega_i \rangle \langle \omega_i | }|\omega_i \rangle = |\omega_i\rangle$, then $[Y_i, J_{Y_i}] = 0$.
\end{enumerate}
\end{cor}
\begin{proof}
Observe that
\begin{equation}
    J_{Y_i} = \begin{pmatrix}
        0 & 0 & 1 \\
        0 & J_{\mathbb I_{x_i}}\otimes  J_{Y_{i-1}} & 0 \\
        1 & 0 & 0  \\
    \end{pmatrix}.
\end{equation}
Then,

\begin{equation}
    J_{Y_i} Y_i J_{Y_i} =
    \begin{pmatrix}
        0 & ( \langle \omega_i | J_{\mathbb I_{x_i}} )^*\otimes \langle f_{Y_{i-1}} |  & 0 \\
       J_{\mathbb{I}_{x_i}  }|\omega_i\rangle ^* \otimes| f_{Y_{i-1}}\rangle & \mathbb I_{x_i} \otimes  J_{Y_{i-1}} Y_{i-1} J_{Y_{i-1}} & J_{\mathbb{I}_{x_i} }|\omega_i\rangle \otimes | l_{Y_{i-1}} \rangle \\
        0 &\langle \omega_i | J_{\mathbb I_{x_i }} \otimes \langle l_{Y_{i-1}} |.
    \end{pmatrix}
\end{equation}
Thus, the second point holds.

    Note that any matrix commutes with its adjugate.
    Suppose $Y_i$ commutes with $J_{Y_i}$. Then,   $J_{Y_i}$ also commutes with $\Lambda_i$ since every matrix commutes with the identity.
    Further,
    \begin{equation}
        J_{Y_i} \Lambda_i J_{Y_i} \Lambda_i^{\mathrm A} = \Lambda_i^{\mathrm A} J_{Y_i} \Lambda_i J_{Y_i}  = C(Y_i;\lambda) \mathbb I.
    \end{equation}
    Since $J_{Y_i}$ is an involution, we then have that
    \begin{equation}
        [J_{Y_i} \Lambda_i^{\mathrm A} J_{Y_i}, \Lambda_i] = 0,
    \end{equation}
    and finally
    \begin{equation}
        (J_{Y_i} \Lambda_i^{\mathrm A} J_{Y_i} ) \Lambda_i = C(Y_i;\lambda) \mathbb I.
    \end{equation}
    Therefore $J_{Y_i} \Lambda_i^{\mathrm A} J_{Y_i}$ is the adjugate of $\Lambda_i$. In other words,
    \begin{equation}\label{eqn:adjcon}
        J_{Y_i} \Lambda_i^{\mathrm A} J_{Y_i} = \Lambda_i^{\mathrm A}.
    \end{equation}
    The rest of the result follows from (\ref{eqn:adjcon}) together with the observation that  $J_X |f_X \rangle = |l_X\rangle$.
\end{proof}
It is straightforward to verify that the choice of gauge espoused in Section \ref{sec:gauge} implies that our glued tree indeed commutes with the exchange matrix of the appropriate size.
\begin{lem}\label{lemma:big}
Adopt the hypothesis and notation of Lemma \ref{lemma:huge}.
    Then,
    \begin{equation}\label{eqn:g}
    \begin{aligned}
        C(Y_i;\lambda) C(Y_{i-1};\lambda)^{x_i} = C(Y_{i-1};\lambda)^2 \left[\langle f_{Y_{i-1}}| R(Y_{i-1};\lambda) |f_{Y_{i-1}}\rangle \langle l_{Y_{i-1}} | R(Y_{i-1};\lambda)|l_{Y_{i-1}} \rangle \right. \\
        \left.- \langle f_{Y_{i-1}} | R(Y_{i-1};\lambda) | l_{Y_{i-1}} \rangle \langle l_{Y_{i-1}} |  R(Y_{i-1};\lambda) | f_{Y_{i-1}}\rangle\right].
        \end{aligned}
    \end{equation}
\end{lem}
\begin{proof}
    Again, we proceed via direct calculation.
    \begin{equation}
        C(Y_i;\lambda) = \det
       \begin{pmatrix}
           -\lambda & (\langle \omega_i | \otimes \langle f_{Y_{i-1}} |)^*   & 0 \\
          (|\omega_i \rangle \otimes | f_{Y_{i-1}}\rangle )^* & \mathbb I_{x_i} \otimes \Lambda_{i-1} & |\omega_i \rangle \otimes | l_{Y_{i-1}}\rangle \\
          0 & \langle \omega_i | \otimes \langle l_{Y_{i-1}} | & -\lambda
       \end{pmatrix} =
       \det
       \begin{pmatrix}
           -\lambda & 0 & \langle \omega_i|^* \otimes \langle f_{Y_{i-1}} | \\
           0 &- \lambda & \langle \omega_i| \otimes \langle l_{Y_{i-1}} |  \\
           | \omega_i^* \rangle \otimes |f_{Y_{i-1}} \rangle & | \omega_i \rangle \otimes |l_{Y_{i-1}} \rangle & \mathbb I_{p \times p} \otimes \Lambda_{i-1}.
       \end{pmatrix}
    \end{equation}
    Now, using (\ref{eqn:diag}), and (\ref{eqn:offdiag}), we have
    \begin{equation}
        C(Y_i;\lambda)=
        \Gamma \det \left[
        \begin{pmatrix} -\lambda & 0 \\ 0 & -\lambda\end{pmatrix} -
        \frac{1}{\Gamma}
        \begin{pmatrix}
           \bra{\omega_i}^*\otimes\bra{f_{Y_{i-1}}} (\mathbb I_{x_i}\otimes \Lambda_{i-1})^{\mathrm A} \ket{\omega_i}^*\otimes \ket{f_{Y_{i-1}}} &  \bra{\omega_i}^*\otimes\bra{f_{Y_{i-1}}} (\mathbb I_{x_i}\otimes \Lambda_{i-1})^{\mathrm A} \ket{\omega_i}\otimes \ket{l_{Y_{i-1}}} \\
            \bra{\omega_i}\otimes\bra{l_{Y_{i-1}}} (\mathbb I_{x_i}\otimes \Lambda_{i-1})^{\mathrm A} \ket{\omega_i}^*\otimes \ket{f_{Y_{i-1}}} &
            \bra{\omega_i}\otimes\bra{l_{Y_{i-1}}} (\mathbb I_{x_i}\otimes \Lambda_{Y_{i-1}})^{\mathrm A} \ket{\omega_i}\otimes \ket{l_{Y_{i-1}}}
        \end{pmatrix}
        \right]
    \end{equation}
    with $\Gamma = \det (\mathbb I_{x_i} \otimes \Lambda_{Y_{i-1}}) = C(Y_{i-1};\lambda)^{x_i}$.
    After using multilinearity of determinants, our result is complete.
\end{proof}

Now we have all of the ingredients necessary to start analyzing the spectrum of glued trees. The recurrence relation therein is most simply written in the following terms: define
\begin{equation}
    \begin{aligned}
        Y_0 &= 0 \\
        \nu_i &=  \frac{1}{\prod_{j = 0}^i C(Y_j;\lambda)^{x_j-1}}\\
        \delta_i &= C(Y_i;\lambda) \nu_i \\
        \phi_i  &= \delta_i \langle f_{Y_i}  |  R(Y_i ;\lambda ) | f_{Y_i} \rangle \\
        \chi_i &= \delta_i \langle f_{Y_i}  | R(Y_i ;\lambda ) | l_{Y_i} \rangle,
    \end{aligned}
\end{equation}
and observe that (\ref{eqn:claims}) may be used to rewrite (\ref{eqn:g}) as
\begin{equation}\label{eqn:nonlinrecur}
    \delta_i \delta_{i-1} = \phi_i^2 - \chi_i^2
\end{equation}
and (\ref{eqn:claims}) itself can be rewritten as
\begin{equation}\label{eqn:linrecur}
\begin{aligned}
    \phi_i = -\lambda \delta_{i-1} - x_i \phi_{i-1} \\
    \chi_i = \sigma (\langle \omega_i | )^* |\omega_i \rangle  \chi_{i-1}
    \end{aligned}
\end{equation}
with, again $\sigma \in \{-1,1\}$. In practice the sign of $\sigma$ is unimportant since it is squared in (\ref{eqn:nonlinrecur}).
When (\ref{eqn:nonlinrecur}) and (\ref{eqn:linrecur}) admit a simple solution, it is possible to construct the characteristic polynomial of $Y_d = A^X$.
We have been able to identify two situations where this is possible: the first is where $\Phi = 0$ or equivalently where $\langle \omega_i | (|\omega_i \rangle)^* = x_i$. The second is the case where $\chi_i = 0$ for all $i$. This occurs at $\Phi = \frac{2 \pi}{x_1}$. These two situations lead to Theorems \ref{thm:fluxless} and \ref{thm:fluxaf} respectively.
We address these two cases in turn.
\subsection{$\Phi = 0$}
The system of recurrence relations of interest to us here is given by
\begin{equation}\label{eqn:simplerec1}
\begin{aligned}
    \delta_i \delta_{i-1} &= \phi_i^2 - \chi_i^2 \\
    \phi_i &= - \lambda \delta_{i-1} - x_i \phi_{i-1} \\
    \chi_i &= x_i \chi_{i-1}.
\end{aligned}
\end{equation}
Use the second and third line of (\ref{eqn:simplerec1}) to rewrite the first as
\begin{equation}\label{eqn:rsolve}
    \delta_i \delta_{i-1} =(\lambda \delta_{i-1} + x_i \phi_{i-1})^2 - x_{i}^2 \chi_{i-1}^2  = \lambda^2 \delta_{i-1}^2 + 2 \lambda x_i \phi_{i-1} \delta_{i-1} + x_i^2 (\phi_{i-1}^2 - \chi_{i-1}^2).
\end{equation}
The last term in (\ref{eqn:rsolve}) is recognized as the first line of (\ref{eqn:simplerec1}) at $i -1$, so we may cancel the common factor of $\delta_{i-1}$ in all terms and write
\begin{equation}
    \delta_i =  \lambda^2 \delta_{i-1} + 2 \lambda x_i \phi_{i-1} + x_i^2 \delta_{i-2} = \lambda( \lambda \delta_{i-1} + x_i \phi_{i-1} ) + x_i (\lambda \phi_{i-1} + x_i \delta_{i-2})
\end{equation}
Evidently, if
\begin{equation}
    \delta_{i-1} = - \lambda \phi_{i-1} - x_i \delta_{i-2},
\end{equation}
then
\begin{equation}
    \delta_i = - \lambda \phi_i - x_i \delta_{i-1}.
\end{equation}
Since this can be imposed by fixing the initial conditions of the recurrence relation, suppose that this is true, and define the sequence of polynomials $\gamma_i$ such that
\begin{equation}
    \gamma_i = \begin{cases}
        \delta_j & i = 2j + 1 \\
        \phi_j  & i = 2j.
    \end{cases}
\end{equation}
Evidently, $\gamma_i $ satisfies
\begin{equation}\label{eqn:simp}
\begin{aligned}
    \gamma_i &= - \lambda \gamma_{i-1} - x_{\lfloor i/2 \rfloor} \gamma_{i-2}, \\
    \gamma_0 &= 1, \\
    \gamma_{-1} &= 0 \\
    \end{aligned}
\end{equation}
which we recognize immediately as the continuant\footnote{We leave it to the reader to verify that the initial conditions in (\ref{eqn:simp}) are correct.}, a recursive relation governing the characteristic polynomial of a tridiagonal matrix.
It is then possible to express $C(Y_i;\lambda)$ as a product of $\gamma_j$ for $j = 0$ to $j = i$. Doing so immediately yields Theorem \ref{thm:fluxless}.

\subsection{$\Phi = \frac{2 \pi}{x_1}$}
In this case, all of the $\chi_i$ variables are identically zero since $\langle \omega_1 | (|\omega_1 \rangle )^* = 0$. In this case, the system of recurrence relations is given by
\begin{equation}
    \begin{aligned}
        \delta_i \delta_{i-1} &= \phi_i^2 \\
        \phi_i &= - \lambda \delta_{i-1} - x_i \phi_{i-1}.
    \end{aligned}
\end{equation}
Solving this system of recurrence relations is much easier: simply note that $\phi_i = \sqrt{\delta_i \delta_{i-1}}$ implies that
\begin{equation}
    \sqrt{\delta_i\delta_{i-1}} = - \lambda \delta_{i-1} - x_i \sqrt{\delta_{i-1} \delta_{i-2}} \iff \sqrt{\delta_i} = - \lambda \sqrt{\delta_{i-1}} - x_i \sqrt{\delta_{i-2}}
\end{equation}
Then, defining
\begin{equation}
    \gamma_i = \sqrt{\delta_i},
\end{equation}
leads to the same continuant relation as above. Proceeding to write $C(Y_i;\lambda)$ in terms of $\delta_j$ (and in turn $\gamma_j^2$)  leads immediately to Theorem \ref{thm:fluxaf}.

\subsection{$p$-nary trees }
For $p$-nary trees, we can proceed even further since the continuants that arise in this case are of tridiagonal matrices that are also Toeplitz. Such matrices are known to have closed form eigenvalues, so one can produce eigenvalue spectra of glued $p$-nary trees without even needing to diagonalize a tridiagonal matrix.
\begin{prop}\label{obs:poly}
    Consider the polynomial defined by the recurrence relation
    \begin{equation}\label{eqn:eigsrecurr}
        \begin{aligned}
            \gamma_n &= -\lambda \gamma_{n-1} - p \gamma_{n-2} \\
            \gamma_0 &= 1\\
            \gamma_{-1} &= 0.
        \end{aligned}
    \end{equation}
    The roots of $\gamma_n$ are
    \begin{equation}
        \lambda = \pm 2 \sqrt{p} \cos(\frac{\pi x}{n+1})
    \end{equation}
    for $x = 1,2,\dots,n$.
    Furthermore, note that
    \begin{equation}
        \left. \frac{1}{m !} \frac{\partial^m}{\partial \lambda^m } \gamma_n \right|_{\lambda = 0} = \begin{cases}
            0 & n + m \text{ odd} \\
            (-p)^{\frac{n - m}{2}} \binom{\frac{n + m}{2}}{\frac{n - m}{2}} & n + m \text{ even}
        \end{cases}
    \end{equation}

\end{prop}
\begin{proof}
    It has long been known \cite{DAFONSECA20017}, that the recurrence relation (\ref{eqn:eigsrecurr}) describes the determinant of the $n$ by $n$ matrix $H$ with elements
    \begin{equation}
        H_{ij} = \begin{cases}
            \lambda & i = j \\
            \sqrt{p} & |i - j| = 1 \\
            0 & |i - j| > 1
        \end{cases}
    \end{equation}
    Furthermore, by direct solution,
    \begin{equation}
    \begin{aligned}
        \gamma_n &= \frac{1}{2^{n+1}\sqrt{\lambda^2 - 4 p}}  \left( \left[\lambda + \sqrt{\lambda^2 - 4p} \right]^{n+1} - \left[ \lambda - \sqrt{\lambda^2 - 4 p} \right]^{n+1} \right) \\
        &= \sum_{l = 0}^{\lfloor \frac{n}{2} \rfloor} \binom{n-l}{l}(-p)^l \lambda^{n - 2 l }.
        \end{aligned}
    \end{equation}
    Note that we have applied the binomial theorem twice.
\end{proof}

    \section{Proof of flat bands for one-dimensional chains from characteristic polynomial recurrence relations}\label{app:chrpolyafb}
    Let $X = \{x_1, x_2, \dots, x_d\}$ be a sequence of positive integers with $x_d > 1$.  Let $X_i = \{ x_j \,:\, j \leq i \}$.
    Just as we did in Appendix \ref{app:chrpolyrec}, we will use the following shorthand notation:
    \begin{equation}
        \begin{aligned}
           Y_i := A^{X_i}(\Phi) \\
           \Lambda_i := Y_i - \lambda.
        \end{aligned}
    \end{equation}
    We consider the lattice consisting of glued trees chained together by their roots.
    By exploiting translation invariance, it is possible to compute the entries of the corresponding complex-weighted adjacency matrix in the block diagonal space consisting of states with a particular translation eigenvalue.
    \begin{defn}[Bloch Hamiltonian]\label{defn:latticehopping}
   Take $k \in [0,2 \pi)$  and define
   \begin{equation}
       \mathcal F(k,\Phi) = \begin{pmatrix}
           0 & \bra{\omega_d}^*\otimes \bra{f_{Y_{d-1}}} + e^{\mathrm i k}\bra{\omega_d}\otimes\bra{l_{Y_{d-1}}} \\
           \ket{\omega_d}^*\otimes\ket{f_{Y_{d-1}}} + e^{-\mathrm i k}\ket{\omega_d}\otimes\ket{l_{Y_{d-1}}} & \mathbb I_{x_d} \otimes Y_{d-1}
       \end{pmatrix}.
   \end{equation}
   We say that $\mathcal F(k,\Phi)$ is the \textbf{Bloch Hamiltonian} of $Y_d = A^X(\Phi)$.
\end{defn}
Precisely, $\mathcal F(k,\Phi)$ is a single block of the tight-binding hopping matrix associated with replacing each link on an infinite one-dimensional chain with a glued tree whose adjacency matrix is $A^X(\Phi)$.

\begin{lem}\label{lem:ham}
Let $\mathcal F(k,\Phi)$ be the Bloch Hamiltonian of some glued tree. We suppress the dependence of $\mathcal F$  on $k$ and $\Phi$ for the sake of cleanliness.
The characteristic polynomial of $\mathcal F$ satisfies
\begin{equation}
    C(\mathcal F;\lambda)  =  \lambda C(Y_{d-1};\lambda)^{x_d} + \bra{f_{Y_d}}\Lambda_d^{\mathrm A}\ket{f_{Y_d}} + \bra{l_{Y_d}} \Lambda_d^{\mathrm A} \ket{l_{Y_d}} + e^{\mathrm i k} \bra{l_{Y_d}}\Lambda_d^{\mathrm A}\ket{f_{Y_d}} + e^{-\mathrm  i k} \bra{f_{Y_d}} \Lambda_d^{\mathrm A} \ket{l_{Y_d}}.
\end{equation}
Furthermore, if
\begin{equation}
    [\mathcal F(Y_d;k), J_{\mathcal F(Y_d;k)}] = 0,
\end{equation}
(as is the case in the canonical gauge) then, by Corollary \ref{cor:symm},
\begin{equation}
    C(\mathcal F;\lambda) =  \lambda C(Y_{d-1};\lambda)^{x_d} + 2 \langle f_{Y_d}| \Lambda_d^{\mathrm A} |f_{Y_d} \rangle + 2\cos(k)\langle l_{Y_d} | \Lambda_d^{\mathrm A} |f_{Y_d} \rangle
\end{equation}
\end{lem}
\begin{proof}
    Use the multilinearity of determinants along with Lemma \ref{lemma:huge} and its intermediate steps.

   \begin{equation}\label{eqn:mutlilinear}
   \begin{aligned}
       \det\Bigl[\mathcal F(k,\Phi) - \lambda \Bigr] &= \det\begin{pmatrix}
           -\lambda & \bra{\omega_d}^*\otimes \bra{f_{Y_{d-1}}} + e^{\mathrm i k}\bra{\omega_d}\otimes\bra{l_{Y_{d-1}}} \\
           \ket{\omega_d}^*\otimes\ket{f_{Y_{d-1}}} + e^{-\mathrm i k}\ket{\omega_d}\otimes\ket{l_{Y_{d-1}}} & \mathbb I_{x_d} \otimes \Lambda_{d-1}
       \end{pmatrix} \\
       &=\det \begin{pmatrix}
           -\lambda & \bra{\omega_d}^*\otimes \bra{f_{Y_{d-1}}}  \\
           \ket{\omega_d}^*\otimes\ket{f_{Y_{d-1}}}  & \mathbb I_{x_d} \otimes \Lambda_{d-1}
       \end{pmatrix} +
       \det\begin{pmatrix}
           0 & \bra{\omega_d}\otimes\bra{l_{Y_{d-1}}} \\
           \ket{\omega_d}\otimes\ket{l_{Y_{d-1}}} & \mathbb I_{x_d} \otimes \Lambda_{d-1}
       \end{pmatrix}  \\
       &+e^{\mathrm i k}\det \begin{pmatrix}
            0& \bra{\omega_d}\otimes\bra{l_{Y_{d-1}}} \\
           \ket{\omega_d}^*\otimes\ket{f_{Y_{d-1}}}  & \mathbb I_{x_d} \otimes \Lambda_{d-1}
       \end{pmatrix} +
       e^{-\mathrm i k}
       \det
       \begin{pmatrix}
           0 & \bra{\omega_d}^*\otimes \bra{f_{Y_{d-1}}}  \\
           \ket{\omega_d}\otimes\ket{l_{Y_{d-1}}} & \mathbb I_{x_d} \otimes \Lambda_{d-1}
       \end{pmatrix}   \\
           \end{aligned}
   \end{equation}
   by the multilinearity of determinants.
   The first term in (\ref{eqn:mutlilinear}) can be recognized as (\ref{eqn:diag})
   and likewise (\ref{eqn:diag}) implies
   \begin{equation}
       \begin{aligned}
       \det\begin{pmatrix}
           0 & \bra{\omega_d}\otimes\bra{l_{Y_{d-1}}} \\
           \ket{\omega_d}\otimes\ket{l_{Y_{d-1}}} & \mathbb I_{x_d} \otimes \Lambda_{d-1}
       \end{pmatrix}   &= -
      \langle \omega_d | \otimes \langle l_{Y_{d-1}} | (\mathbb I_{x_d} \otimes \Lambda_{d-1})^{\mathrm A} |\omega_d \rangle \otimes |l_{Y_{d-1}} \rangle  \\
      &=   \langle f_{Y_d} | \Lambda_d^{\mathrm A} | f_{Y_d} \rangle + \lambda C(Y_{d-1} ;\lambda)^{x_d}.
      \end{aligned}
   \end{equation}
   Finally, the second line of (\ref{eqn:mutlilinear}) can be rewritten in the desired form by comparison with (\ref{eqn:offdiag}).
\end{proof}
    Lemma \ref{lem:travel} then immediately implies that
    \begin{equation}
        \frac{\mathrm d}{\mathrm d k} C(\mathcal F(k,\Phi);\lambda) = 0
    \end{equation}
    for all values of $\lambda$ so long as $\Phi \in F^X$, which is sufficient to show that all of the eigenvalue bands of $\mathcal F$ are independent of $k$.

    \section{Glued tree spectra from eigenstate construction}\label{app:chrpolystat}

    We now examine the eigenvectors of glued tree adjacency matrices.
    Let $X = \{x_1, x_2, \dots, x_d\}$ be a set of positive integers with $x_d > 1$. We concern ourselves with the eigenvectors of $A^X$.
    We produce the eigenvectors of $X$ by using Gram--Schmidt orthogonalization.
Define the states
\begin{equation}
\begin{aligned}
    |\varphi_0 \rangle &= |f_{A^X}\rangle \\
    |\varphi_i \rangle &= \frac{1}{\sqrt{x_i} }\left(\mathbb I - \sum_{j = 0}^{i-1} |\varphi_j \rangle \langle \varphi_j | \right) A^X |\varphi_{i-1} \rangle.
\end{aligned}
\end{equation}
The observant reader may notice that
\begin{equation}
    |\varphi_i\rangle = \frac{1}{\left(\prod_{j = 1}^{i} x_i\right)^{\frac{1}{2}}}\sum_{\text{dist}(l,0) = i} |l \rangle.
\end{equation}
where $\text{dist}(u,v)$ denotes graph distance in $\mathcal G(A^X)$.
Furthermore,
\begin{equation}
    A^X |\varphi_i \rangle  = \sqrt{x_i} |\varphi_{i-1} \rangle + \sqrt{x_{i+1}} |\varphi_{i+1}\rangle
\end{equation}
for $i = 1,2,\dots, d-1$.
For $i = d+1, d + 2,\dots 2d - 1$,
\begin{equation}
    A^X |\varphi_i \rangle  = \sqrt{x_{ 2d - i + 1}} |\varphi_{i-1} \rangle + \sqrt{x_{2d - i}} |\varphi_{i+1}\rangle
\end{equation}
For the remaining values of $i$,
\begin{equation}
\begin{aligned}
    A^X|\varphi_0 \rangle &= \sqrt{x_1} |\varphi_1 \rangle  \\
    A^X|\varphi_{d} \rangle &= \sqrt{x_d} |\varphi_{d-1}\rangle + \sqrt{x_d} |\varphi_{d+1}\rangle\\
    A^X|\varphi_{2d} \rangle &= \sqrt{x_1} |\varphi_{2d - 1} \rangle  \\
    \end{aligned}
\end{equation}
The Hilbert space of vectors spanned by $|\varphi_i\rangle$ states is thus closed.

Then, choose $|\varphi^m_1\rangle$ for $m = 1, 2, \dots , x_1-1$ to be an orthonormal set of vectors that spans the complement of $|\varphi_1\rangle$. Define
\begin{equation}
    |\varphi_i^m \rangle = \frac{1}{\sqrt{x_i} }\left(\mathbb I - \sum_{j = 0}^{i-1} |\varphi_j^m \rangle \langle \varphi_j^m | \right) A^X |\varphi_{i-1}^m \rangle.
\end{equation}
for every such vector, yielding a further $x_1 - 1$ Hilbert spaces closed under the application of $A^X$.
Repeating this procedure until it no longer produces new ``nice" subspaces provides a unitary $U$ satisfying Theorem \ref{thm:fluxless}.

This technique of identifying Hilbert spaces closed under application of a matrix that are of dimension linear in the depth of a glued tree is possible for glued trees at $\Phi = 0$. This technique, sadly, is not effective at producing spectra for complex-weighted glued trees because interference effects break many of the symmetries present in $A^X$ (namely various permutation symmetries corresponding to the freedom to draw children of a given node in any order).
\end{appendix}

\section{Data Availability}
The authors declare that the data supporting the findings of this study are available upon reasonable request to the corresponding author.

\section{Conflict of Interest}
Princeton University Professor Andrew Houck is also
a consultant for Quantum Circuits Incorporated (QCI).
Due to his income from QCI, Princeton University has
a management plan in place to mitigate a potential conflict of interest that could affect the design, conduct and
reporting of this research.
The authors declare no other conflicts of interest.

\bibliography{bibliography}
\end{document}